%% file: main.tex
\documentclass[dvipsnames,11pt]{article}
\usepackage{amsmath,amssymb,amsthm,color,bm,mathrsfs,extarrows,tikz,graphicx,mathtools,enumitem,fancybox,makecell,txfonts}
\usepackage[square,sort,comma,numbers]{natbib}
\usepackage{subcaption}
\usepackage[ruled]{algorithm2e}
\usepackage[margin=1.in]{geometry}
\usepackage{xcolor,bbm}
\usepackage[pagebackref]{hyperref}
\usepackage{qcircuit-latest}

\allowdisplaybreaks
\usetikzlibrary{decorations.pathreplacing,decorations.markings,decorations.pathmorphing,decorations.shapes,arrows.meta,positioning,patterns}

\usepackage[capitalise,nameinlink]{cleveref}
\Crefname{lemma}{Lemma}{Lemmas}
\Crefname{fact}{Fact}{Facts}
\Crefname{theorem}{Theorem}{Theorems}
\Crefname{corollary}{Corollary}{Corollaries}
\Crefname{claim}{Claim}{Claims}
\Crefname{example}{Example}{Examples}
\Crefname{problem}{Problem}{Problems}
\Crefname{definition}{Definition}{Definitions}
\Crefname{notation}{Notation}{Notations}
\Crefname{assumption}{Assumption}{Assumptions}
\Crefname{subsection}{Subsection}{Subsections}
\Crefname{section}{Section}{Sections}
\Crefformat{equation}{(#2#1#3)}

\newtheorem{theorem}{Theorem}[section]
\newtheorem*{theorem*}{Theorem}

\newtheorem{proposition}[theorem]{Proposition}
\newtheorem*{proposition*}{Proposition}
\newtheorem{lemma}[theorem]{Lemma}
\newtheorem*{lemma*}{Lemma}
\newtheorem{corollary}[theorem]{Corollary}
\newtheorem*{corollary*}{Corollary}
\newtheorem*{conjecture*}{Conjecture}

\newtheorem*{fact*}{Fact}

\newtheorem*{exercise*}{Exercise}

\newtheorem*{hypothesis*}{Hypothesis}

\theoremstyle{definition}
\newtheorem{definition}[theorem]{Definition}

\newtheorem{problem}[theorem]{Problem}

\newtheorem{exercise-easy}[theorem]{Exercise}
\newtheorem{exercise-med}[theorem]{Exercise}
\newtheorem{exercise-hard}[theorem]{Exercise$^\star$}
\newtheorem{claim}[theorem]{Claim}
\newtheorem*{claim*}{Claim}

\newtheorem*{remark*}{Remark}

\newtheorem*{observation*}{Observation}

\DeclareMathAccent{\widehat}{\mathord}{largesymbols}{"62}

\DeclareMathOperator{\E}{\mathbb E}

\DeclareMathOperator{\rank}{rank}

\DeclareMathOperator{\poly}{poly}

\newcommand{\vabs}[1]{\left\| #1 \right\|}
\newcommand{\pbra}[1]{\left( #1 \right)}
\newcommand{\cbra}[1]{\left\{ #1 \right\}}
\newcommand{\sbra}[1]{\left[ #1 \right]}

\newcommand{\ceilbra}[1]{\left\lceil #1 \right\rceil}
\newcommand{\floorbra}[1]{\left\lfloor #1 \right\rfloor}
\newcommand{\bmat}[1]{\begin{bmatrix} #1 \end{bmatrix}}
\newcommand{\bin}{\{0,1\}}
\newcommand{\GL}{\mathrm{GL}}
\renewcommand{\P}{\mathcal{P}}
\newcommand{\NP}{\mathcal{NP}}
\newcommand{\rBSC}{r\mathtt{BSC}}
\newcommand{\GCM}{\mathtt{GCM}}
\newcommand{\LSM}{\mathtt{LSM}}

\newcommand{\Cbb}{\mathbb{C}}
\newcommand{\Fbb}{\mathbb{F}}

\newcommand{\Ccal}{\mathcal{C}}
\newcommand{\Ecal}{\mathcal{E}}

\newcommand{\Ical}{\mathcal{I}}
\newcommand{\Scal}{\mathcal{S}}
\newcommand{\Tcal}{\mathcal{T}}
\newcommand{\Gscr}{\mathscr{G}}
\newcommand{\Lscr}{\mathscr{L}}
\newcommand{\Amat}{\bm{\mathrm A}}
\newcommand{\Bmat}{\bm{\mathrm B}}
\newcommand{\Cmat}{\bm{\mathrm C}}

\newcommand{\Lmat}{\bm{\mathrm L}}
\newcommand{\Pmat}{\bm{\mathrm P}}
\newcommand{\Rmat}{\bm{\mathrm R}}
\newcommand{\Tmat}{\bm{\mathrm T}}
\newcommand{\Mmat}{\bm{\mathrm M}}
\newcommand{\Emat}{\bm{\mathrm E}}
\newcommand{\Fmat}{\bm{\mathrm F}}
\newcommand{\Imat}{\bm{\mathrm I}}
\newcommand{\Umat}{\bm{\mathrm U}}
\newcommand{\Vmat}{\bm{\mathrm V}}

\newcommand{\Xmat}{\bm{\mathrm X}}
\newcommand{\Ymat}{\bm{\mathrm Y}}

\newcommand{\onemat}{\bm{\mathrm 1}}

\newcommand{\Rsfmat}{\bm{\mathsf R}}

\newcommand{\ket}[1]{\left| #1 \right\rangle}

\title{Optimal Space-Depth Trade-Off of CNOT Circuits in Quantum Logic Synthesis}

\author{
Jiaqing Jiang\thanks{California Institute of Technology, United States. Email: \texttt{jiaqingjiang95@gmail.com}}
\and
Xiaoming Sun\thanks{Institute of Computing Technology, Chinese Academy of Science, China. University of Chinese Academy of Sciences, China. Email: \texttt{\{sunxiaoming, zhangjialin\}@ict.ac.cn}}
\and
Shang-Hua Teng\thanks{University of Southern California, United States. Email: \texttt{shanghua@usc.edu}}
\and
Bujiao Wu\thanks{Peking University, China. Email: \texttt{bujiaowu@gmail.com}}
\and
Kewen Wu\thanks{University of California, Berkeley, United States. Email: \texttt{shlw\_kevin@hotmail.com}}
\and
Jialin Zhang\footnotemark[2]
}

\date{}

\begin{document}
\maketitle 

\begin{abstract}
Decoherence---in the current physical implementations of quantum computers---makes {\em depth reduction} a vital 
task in quantum-circuit design. 
Moore and Nilsson (SIAM Journal of Computing, 2001) demonstrated that additional qubits--known as {\em ancillae}---can be used to provide an extended {\em space} to {\em parallelize} quantum circuits.
Specifically, they proved that, with $O(n^2)$ ancillae, any $n$-qubit CNOT circuit can be transformed into an equivalent one of $O(\log n)$ depth.
However, the near-term quantum technologies can only support a limited amount of qubits, making {\em space-depth trade-off} a fundamental research subject for quantum-circuit synthesis.
 
In this work, we establish an asymptotically optimal space-depth trade-off for CNOT circuits.
We prove that any $n$-qubit CNOT circuit can be parallelized to $O\left(\max\left\{\log n, \frac{n^2}{(n+m)\log (n+m)}\right\}\right)$ depth with $m$ ancillae.
This bound is tight even if the task is expanded from exact synthesis to the approximation of CNOT circuits with arbitrary two-qubit quantum gates.
Our result can be extended to stabilizer circuits via the reduction by Aaronson and Gottesman (Physical Review A, 2004).
Furthermore, we provide hardness evidence for optimizing CNOT circuits in terms of size or depth.

Our result has improved upon two previous papers that motivated our work.
\begin{itemize}
\item {\em Moore-Nilsson's construction (aforementioned) for $O(\log n)$-depth CNOT circuit synthesis}: 
We have reduced their need for ancillae by a factor of $\log^2 n$ by showing that $m= O(n^2/\log^2 n)$ additional qubits---which is asymptotically optimal--- suffice to build equivalent $O(\log n)$-depth $O(n^2/\log n)$-size CNOT circuits.
\item {\em Patel-Markov-Hayes's construction (Quantum Information \& Computation 2008) for $m = 0$}:
We have reduced their depth by a factor of $n$ and achieved the asymptotically optimal bound of $O(n/\log n)$.
\end{itemize}
\end{abstract}

\section{Introduction}
\label{sec:introduction}

{\em Quantum-circuit synthesis} is a fundamental task in quantum computing.
Given any $n$-qubit unitary operator, it implements the operator as a sequence of low-level gates, aiming to optimize the {\em size} and/or {\em depth} of the designed circuit~\cite{barenco1995elementary,divincenzo1995two}.
During the last quarter century, quantum-synthesis algorithms have been developed to achieve asymptotically optimal size~\cite{ShendeBM06,vartiainen2004efficient,knill1995approximation,shende2004smaller}.

Depth reduction---a vital step for combating decoherence in the current physical implementations of quantum computing systems---has been a challenging problem.
In order to significantly shorten the critical-paths, synthesis methods commonly use ancillae---i.e., additional qubits---to enhance {\em parallelism}. 
For example, with sufficient ancillae, the Quantum Fourier Transform can be approximated by an $O(\log n+\log\log(1/\epsilon))$-depth circuit~\cite{cleve2000fast}, and the stabilizer circuits can be parallelized to $O(\log n)$ depth~\cite{moore2001parallel}. 
Depth optimizations of Clifford+T circuits are presented in \cite{selinger2013quantum,amy2014polynomial}.

However, the near-term quantum devices only have a small number of qubits~\cite{preskill2018quantum}, fundamentally 
limiting the amount of ancillae,
needed as ``auxiliary logic space'' for depth optimization.
This practical concern gives rise to the following basic {\em space-depth trade-off problem} in quantum-circuit synthesis: 
\begin{center}
{\em What is the best achievable depth for each given number of ancillae?}
\end{center}

\subsection{CNOT-Circuit Optimization}

Since controlled-NOT gate (CNOT) and single-qubit operations form a {\em universal set} for quantum computing~\cite{barenco1995elementary,divincenzo1995two}, CNOT-circuit optimization---addressing both circuit size and circuit depth---has been extensively studied in quantum-logic synthesis.
Our work also focuses on CNOT circuits.

For size optimization, Patel, Markov, and Hayes~\cite{patel2008optimal} proved that each $n$-qubit CNOT circuit can be synthesized with $O(n^2/\log n)$ CNOT gates, and this bound is {\em asymptotically tight}.
When topological constraints---i.e., limited two-qubit connectivity among addressable qubits---exist, synthesis algorithms have been designed by Kissinger-de Griend \cite{kissinger2020cnot} and Nash-Gheorghiu-Mosca \cite{nash2020quantum} to build circuits of size $O(n^2)$.

For depth optimization, to our knowledge, $O(n^2/\log n)$ (following Patel-Markov-Hayes's size bound~\cite{patel2008optimal}) was also the best general upper bound known for the case without ancillae.
Moore and Nilsson \cite{moore2001parallel} 
effectively utilized ancillae for depth reduction.
They proved that given $O(n^2)$ ancillae, any $n$-qubit CNOT circuit can be parallelized into $O(\log n)$ depth.
Results of CNOT-circuit optimization can be extended to stabilizer circuits.
In \cite{aaronson2004improved}, Aaronson and Gottesman established a strong connection between these two families:
Stabilizer circuits have a canonical form of $11$ blocks, in which each block consists of only one type of gates from gate set $\{\text{CNOT, Phase, Hadamard}\}$.
Since each block of Phase gates or Hadamard gates has depth 1 and size at most $n$, one can transfer CNOT-circuit optimization techniques to stabilizer circuits.

\subsection{Our Main Results}
In this paper, we establish an asymptotically optimal space-depth trade-off for quantum-logic synthesis. 
The heart of our results is the following theorem, characterizing 
the capacity of ancillae in parallelizing CNOT-circuits.

\begin{theorem}[Optimal Quantum Space-Depth Trade-off]\label{thm:upperbound}
For any integer $m\geq 0$, any $n$-qubit CNOT circuit can be transformed, using $m$ ancillae, into an equivalent CNOT circuit with depth
$$
O\pbra{\max\cbra{\log n, \frac{n^2}{(n+m)\log (n+m)}}}.
$$
The synthesis algorithm in the stated construction takes $\widetilde O\pbra{n^\omega}$ time, where $\omega$ is the matrix multiplication exponent~\cite{AlmanW21}.
\end{theorem} 

\Cref{thm:upperbound} can be readily extended to stabilizer circuits, using a reduction designed by Aaronson and Gottesman \cite{aaronson2004improved}. 
Mathematically, 
for any stabilizer circuit, 
there exists an equivalent circuit 
consisting of the following $11$ blocks sequence H-C-P-C-P-C-H-P-C-P-C, where 
H only contains Hadamard gates, 
C only contains CNOT gates, and
and P only contains Phase gates.
Since Hadamard gate or phase gate is single-qubit gate that can be merged, one block of them has at most depth 1 and size $n$. 
Therefore, for depth reduction, it suffices to optimize the CNOT blocks. 
Combining \Cref{thm:upperbound}, we have the following corollary.

\begin{corollary}[Depth-Optimization of Stabilizer Circuits]
For any integer $m\geq 0$, any $n$-qubit stabilizer circuit can be parallelized, with $m$ ancillae, to depth
$$O\pbra{\max\cbra{\log n, \frac{n^2}{(n+m)\log (n+m)}}}.$$
\end{corollary}

Inspired by the two prior landmark results on CNOT-circuit synthesis,
our result in \Cref{thm:upperbound} 
also improves upon both of these earlier constructions:
\begin{itemize}
\item {\bf Depth Optimization with Sufficient Ancillae}: By parallelizing any CNOT circuit into an equivalent one with $O(\log n)$ depth and $O(n^2/\log n)$ size using only $m=O(n^2/\log^2 n)$ ancillae, we reduce the number of ancillae needed by Moore and Nilsson \cite{moore2001parallel} by a factor of $\log^2 n$. 
\item {\bf Depth Optimization without Ancillae}: By achieving the asymptotically optimal depth bound of $O(n/\log n)$ in parallel CNOT-circuit synthesis without any ancillae---i.e., for the case $m = 0$---we reduce the depth implied in the work of Patel, Markov, and Hayes \cite{patel2008optimal} by a factor of $n$.
\end{itemize} 
 
These improvements are the best possible.
We prove---by a counting argument---that our space-depth trade-off is asymptotically tight. 
In fact, the optimality of \Cref{thm:upperbound} is established in a more general setting as stated in \Cref{thm:lowerbound} below.
We show that, even if one is allowed to use arbitrary two-qubit quantum gates---rather than CNOT gates only---to approximately implement given CNOT circuits, the depth of our synthesized circuits still meet the achievable 
lower bound (in the the worst-case). 

\begin{theorem}[Optimality of Quantum-Circuit Depth]\label{thm:lowerbound}
Let $\epsilon<\sqrt2/2$ be a non-negative constant.
For $1-o(1)$ fraction of $n$-qubit CNOT circuits, 
any $\epsilon$-approximate\footnote{In essence, an {\em $\epsilon$-approximate circuit} outputs a quantum state close to the CNOT circuit's output under $\ell_2$ distance.} $n$-qubit $m$-ancilla quantum circuit has depth
$$
\Omega\pbra{\max\cbra{\log n, \frac{n^2}{(n+m)\log (n+m)}}}.
$$
\end{theorem}

Besides the depth, for $m=O(n^2/\log^2 n)$, our construction has size $O(n^2/\log n)$. 
Techniques in \cite{andren2007complexity,christofides2014asymptotic}
can be generalized to show that such $n$-qubit $m$-ancilla circuit must have size $\Omega(n^2/\log n)$. 
Thus our synthesis algorithm is also asymptotically optimal in circuit size.

\subsection{Technical Highlights}

Mathematically, our synthesis method for \Cref{thm:upperbound} is based on carefully-designed Gaussian eliminations.
As observed in \cite{patel2008optimal}, any $n$-qubit CNOT circuit can be represented by an invertible matrix $\Mmat\in\Fbb_2^{n\times n}$, and the synthesis of CNOT circuit is equivalent to transform $\Mmat$ to identity matrix by Gaussian eliminations. 
Our optimal synthesis, aiming to reduce the depth of CNOT circuits, is the product of \emph{parallel Gaussian eliminations}. 
We minimize the number of parallel Gaussian eliminations using the following two techniques:
\begin{itemize}
\item {\bf Elimination Inspired by Oblivious Routing}: 
For the case without any ancilla, we first establish that if the structure of the matrix is {\em near random}, then it is amenable to effective parallel Gaussian elimination.
We then prepare the input circuit by using a popular idea from oblivious routing~\cite{KaklamanisKT91} to ensure the existence of a close-to-random structure.
This randomization step is then derandomized by a standard approach, but with somewhat subtle conditional expectations.
\item {\bf Elimination Inspired by the Method of Four Russians}: 
For the case with non-zero ancillae, we adopt the idea from {\em the Method of Four Russians.} 
In \cite{patel2008optimal}, Patel \emph{et al.} used the Four Russians' algorithm to eliminate $\log n$ columns, namely $n\log n$ elements, by $O(n)$ Gaussian eliminations.
In our work, we deal with $\log^2 n$ columns, namely $n\log^2 n$ elements, by $O(\log n)$ parallel Gaussian eliminations. 
This is done by preparing the additive basis of Boolean vectors \cite{Lovett17} and properly balancing the trade-off between resource and cost.
\end{itemize}

Both our ancilla-free and ancilla-based synthesis algorithms for CNOT circuits rely on the 1-factorization of almost regular bipartite graph \cite{ColeOS01,Alon03}---a direct application of Hall's marriage theorem---to get an ideal ordering. 
We show that both algorithms run in time 
$\widetilde{O}(n^\omega)$, where $\omega$ is the universal constant for matrix multiplication \cite{AlmanW21}.

Our results can be directly extended to matrix decomposition over finite fields. 
More precisely, we show that any $\Mmat\in\GL(n,2)$ can be decomposed to $O(n/\log n)$ parallel row-elimination matrices.
This parallel matrix-decomposition method can be generalized to finite field $\Fbb_q$ for any constant $q$. 

\begin{theorem}[Parallel Matrix Decomposition]\label{thm:FfieldParallel}
For any constant $q$ and $\Mmat\in\GL(n,q)$, 
$\Mmat$ can be transformed to the identity matrix 
by $O(n/\log n)$ parallel Gaussian eliminations.
\end{theorem}

Our construction also indicates that there might be a parallel Gaussian elimination algorithm for solving linear equations over $\Fbb_q$ by $O(n/\log n)$ parallel row-elimination matrices with $O(n/\log n)$ parallel time.
Some related work can be found in \cite{faugere2010parallel,rupp2006parallel,robert1989optimal,cosnard1988parallel}.
We leave this as an open problem for future research.

Note that our space-depth trade-off for quantum-synthesis is optimal only in the asymptotic sense.
A fundamental optimization problem, related to our results for the CNOT-circuit depth reduction, is the construction of an equivalent circuit with optimal depth for any \emph{given} CNOT circuit and amount of ancillae.
Specifically, given any matrix $\Mmat\in\GL(n,2)$ and a pair of integers $k> 0$ an $m \geq 0$, the problem is to determine whether there exists an $n$-qubit $m$-ancilla CNOT circuit for $\Mmat$ with depth at most $k$.
This decision problem is an analogue of the {\sc Minimum Circuit Size} Problem ($\tt MCSP$)---the famous problem that is unlikely to be proven in $\P$ or $\NP$-complete by natural proofs~\cite{kabanets2000circuit,razborov1997natural}. 
In this paper, we provide, what we consider to be, a relevant hardness evidence by proving hardness results for optimizing CNOT circuits in two slightly different scenarios.
\begin{itemize}
\item 
In the first scenario, the objective is to optimize, with ancillae, the depth of a CNOT circuit under certain topological constraints \cite{devitt2016performing,divincenzo2000physical}.
\item 
In the second scenario, the goal is to optimize, with ancillae, a sub-circuit of a CNOT circuit.
\end{itemize}
Below, we give a high-level summary of the inapproximability results in \Cref{thm:complexityTotal}. The formal statement can be found in \Cref{sec:hardnessresult}. 

\begin{theorem}[Inapproximability - Informal]\label{thm:complexityTotal}
It is $\NP$-hard to approximate the solution of the following problems 
 within any constant factor:
\begin{itemize}
\item {\sc Global Constrained Minimization} ($\GCM$): 
Given an $n$-qubit CNOT circuit, an integer $m$, and topological constraints on the $n+m$ qubits, output an equivalent $n$-qubit $m$-ancilla CNOT circuit of minimum size or depth.
\item {\sc Local Size Minimization} ($\LSM$): 
Given an $n$-qubit CNOT circuit, an integer $m$, and a specific part of the circuit, output an equivalent $n$-qubit $m$-ancilla CNOT circuit which optimizes the size of the specified sub-circuit.
\end{itemize}
\end{theorem}

At a high level, {\sc Global Constrained Minimization} ($\GCM$) aims to find the optimal size or depth under certain topological constraints $S$, where CNOT gate with control $i$ and target $j$ is {\em legal} iff $(i,j)\in S$. 
Such a restriction is common in existing quantum devices~\cite{devitt2016performing,divincenzo2000physical}, and 
CNOT-circuit optimization on such devices has been discussed~\cite{nash2020quantum,kissinger2020cnot}. 
The {\sc Local Size Minimization} ($\LSM$) aims to optimize the size of a selected part of the circuit while leaving other parts unchanged. 
Hardness for general quantum circuit optimization over topological constraints can be seen in \cite{herr2017optimization,childs2019circuit}.

\subsection{Paper organization}
In \Cref{sec:preliminary}, we review basic notations and definitions needed for the paper. 
In \Cref{sec:noancilla}, we focus on CNOT-circuit synthesis without any ancilla. 
We also prove that any tree-based CNOT circuit can be parallelized to depth $O(\log^2 n)$ without any ancilla.
In \Cref{sec:withancilla}, we present our ancilla-based synthesis algorithm and complete the proof of \Cref{thm:upperbound}. 
In \Cref{sec:lowernpc}, we give the tight lower bound showing that our space-depth trade-off is asymptotically optimal.
We also present the related hardness results for the quantum-synthesis optimization problems discussed above. 
Finally in \Cref{sec:conclusion}, we summarize the paper and discuss future research directions.

\section{Preliminary}\label{sec:preliminary}

\paragraph*{Basic Notation.} 
We use $[n]$ to denote $\cbra{1,\ldots,n}$, 
$\Cbb$ to denote the complex domain, 
$\Fbb_q$ to denote the finite field with $q$ elements;
for $q=2$, we use $\oplus$ to denote the exclusive-or, i.e., the addition under $\Fbb_2$.

We use $\GL(n,q)$ to denote the set of $n\times n$ invertible matrices with entries from $\Fbb_q$, superscript $\top$ to denote the transpose of matrix or vector,
$\Imat_n$ to denote the $n\times n$ identity matrix (its subscript is omitted if the context is clear),
$\onemat_{p,q},p\neq q$ to denote the all-zero matrix except that the $(p,q)$-th entry equals $1$.

For any matrix $\Mmat$, we use $\Mmat[i,j]$, $\Mmat[i,*]$, $\Mmat[*,j]$ to denote, respectively, the $(i,j)$-th entry, the $i$-th row, and $j$-th column 
in $\Mmat$.

Lastly, we use $\widetilde O$ to hide the polylogarithmic terms in the traditional Big-O notation.

\paragraph*{CNOT Gate and Circuits.} 
A CNOT gate maps Boolean tuple $(x,y)$ to $(x,x\oplus y)$.
Mathematically, it is the invertible linear map 
$\begin{bmatrix}
1&0\\
1&1 
\end{bmatrix}
$
over $\Fbb_2^2$.
Thus, each $n$-qubit CNOT circuit $\Ccal$ can be viewed as an invertible linear map over $\Fbb_2^n$, represented as an invertible matrix $\Mmat$ in $\Fbb_2^{n\times n}$. 

\paragraph*{Ancilla-Based CNOT Circuits.}
An $n$-qubit $m$-ancilla CNOT circuit has $m$ ancillae with initial assignment $\ket{0}^{\otimes m}$, and satisfies the following key property:
After the evaluation of the circuit, all ancillae are restored.
An $n$-qubit $m$-ancilla CNOT circuit $\Ccal$ {\em implements} an invertible matrix $\Mmat\in\Fbb_{2}^{n\times n}$ if for any $\bm x\in \Fbb_{2}^{n}$,
 the output of circuit $\Ccal$ on input $\ket{\bm x}\ket{ 0}^{\otimes m}$ 
 is $\ket{\Mmat \bm x}\ket{ 0}^{\otimes m}$. 
In other words, the matrix representation for $\Ccal$ is
$\begin{bmatrix}
\Mmat&\Emat\\
\bm 0&\Fmat
\end{bmatrix}
$
for some $\Emat\in\Fbb_{2}^{n \times m}$ and invertible $\Fmat\in\Fbb_{2}^{m\times m}$. 
Since $\Emat,\Fmat$ do not interfere the output when the input is $\ket{\bm x}\ket{0}^{\otimes m}$, we abbreviate them as $\bm*$. In particular, $\bm*_m$ represents some matrix in $\Fbb_{2}^{m\times m}$. 
In the rest of this paper, we say such $\Ccal$ is an $n$-qubit $m$-ancilla CNOT circuit for $\Mmat$. 
\paragraph*{Equivalent Ancilla-Based CNOT Circuits.}
We say an $n$-qubit $m_1$-ancilla CNOT circuit $\Ccal_1$ is equivalent to an $n$-qubit $m_2$-ancilla CNOT circuit $\Ccal_2$ (denoted by $\Ccal_1\cong\Ccal_2$) if for any $\bm x\in \bin^n$, the output of the first $n$ qubits are the same for $\Ccal_1\ket{\bm x}\ket{ 0}^{\otimes m_1}$ and $\Ccal_2\ket{\bm x}\ket{ 0}^{\otimes m_2}$. 

\paragraph*{Row-Elimination Matrices.}
Mathematically, a CNOT gate with control qubit $i$ and target qubit $j$ can be represented as a row elimination from $i$ to $j$ (i.e., adding row-$i$ to row-$j$).
Thus, a CNOT circuit can be viewed as the product of a sequence of row-elimination matrices.

\begin{definition}[Row-Elimination Matrix]
We say a matrix $\Rmat\in\Fbb_2^{n\times n}$ is a \emph{row-elimination matrix} if $\Rmat=\Imat$ or there exist $i,j\in[n],i\neq j$, such that
$$
\Rmat=\Imat+\onemat_{j,i}.
$$
\end{definition}
Note that for any row-elimination matrix $\Rmat$, we have $\Rmat^2 = \Imat$.
We use $\Rsfmat(i,j)$ to denote $\Imat+\onemat_{j,i}.$
A row-elimination matrix represents a single basic step in the process of Gaussian elimination. 
For any matrix, left-multiplying by $\Rsfmat(i,j)$ represents adding the $i$-th row to the $j$-th row.

\paragraph*{Parallel Row-Elimination Matrices.}
A basic concept for depth reduction in CNOT-circuit synthesis is parallel row elimination (or equivalently, parallel Gaussian elimination).
\begin{definition}[Parallel Row-Elimination Matrix]
We say matrix $\Rmat\in\Fbb_2^{n\times n}$ is a \emph{parallel row-elimination matrix} if $\Rmat=\Imat$ or there exist $t\in[n]$ and $\bm i,\bm j\in[n]^t$, such that $i_k,j_k$'s are $2t$ different indices, and
$$
\Rmat=\Imat+\sum_{k=1}^t\onemat_{j_k,i_k}=\prod_{k=1}^t\Rsfmat(i_k,j_k).
$$
\end{definition}

A parallel row-elimination matrix represents several independent steps in Gaussian elimination.
Since all $i_k,j_k$ are distinct, there is no need to apply these eliminations in any particular order. 
When $\bm i,\bm j$ is clear in the context (like in \Cref{sec:noancilla} and \Cref{sec:withancilla}), we use $\Rmat$ to denote a parallel row-elimination matrix, and $\Rmat_1,\Rmat_2,\ldots$ to denote a sequence of parallel row-elimination matrices.

\paragraph*{Quantum Approximation.}
In this paper, without loss of generality, we only consider quantum circuits consisting of single-qubit and two-qubit gates.
In \Cref{sec:lowernpc}, we will use the following definitions of quantum-circuit approximation. 

\begin{definition}[$\epsilon$-close]
For any $\epsilon>0$, two vectors $\bm u,\bm v\in\Cbb^{n}$ are $\epsilon$-close iff $\vabs{\bm u-\bm v}_2<\epsilon$.
\end{definition}

\begin{definition}[$\epsilon$-approximate]
Given $n$-qubit quantum circuit $\Ccal_1$ and $n$-qubit $m$-ancilla quantum circuit $\Ccal_2$, we say $\Ccal_2$ $\epsilon$-approximates $\Ccal_1$ if for any $\bm x\in\bin^n$,
\begin{itemize}
\item $\Ccal_2$ maps $\ket{\bm x}\ket{0}^{\otimes m}$ to $\ket{\varphi_{\bm x}}\ket{0}^{\otimes m}$ for some $ \ket{\varphi_{\bm x}}\in\Cbb^{2^n}$ with $\vabs{ \varphi_{\bm x}}_2=1$.
\item $\ket{\varphi_{\bm x}}$ and $\Ccal_1\ket{\bm x}$ are $\epsilon$-close.
\end{itemize}
\end{definition}

\section{CNOT-Circuit Synthesis: Depth Reduction Without Ancillae}\label{sec:noancilla}

We will divide the proof of \Cref{thm:upperbound} into two parts. 
In this section, we focus on the case when $m=o(n)$.
The rest case will be fully addressed in \Cref{sec:withancilla}. 

For the case of $m=o(n)$, it is sufficient to establish the following \Cref{thm:noancillaCNOT}, which provides an $\widetilde{O}\pbra{n^{\omega}}$-time algorithm to transform any $n$-qubit CNOT circuit to an equivalent one with depth $O(n/\log n)$. 
The construction naturally covers the cases in \Cref{thm:upperbound} when there are $o(n)$ ancillae.
\begin{theorem}[Ancilla-Free Synthesis]\label{thm:noancillaCNOT}
There is an $\widetilde O(n^\omega)$-time algorithm to parallelize any $n$-qubit CNOT circuit to depth $O(n/\log n)$ without ancillae.
\end{theorem}
Due to the connection between CNOT circuits and invertible Boolean matrices, we can reformulate \Cref{thm:noancillaCNOT} as follows.
\begin{lemma}[Parallel Gaussian-Eliminatin Reformulation of \Cref{thm:noancillaCNOT}]\label{lem:noancillaMatrix}
There is an $\widetilde O(n^\omega)$-time algorithm for the following task: Given any $\Mmat\in\GL(n,2)$, it outputs parallel row-elimination matrices $\Rmat_1,\ldots,\Rmat_d$ where $d=O(n/\log n)$, such that 
$$
\Rmat_d\cdots\Rmat_2\Rmat_1\Mmat=\Imat.
$$
\end{lemma}
\begin{proof}
First, by Bunch and Hopcroft~\cite{LUdecom}, we can factorize, in time $\widetilde O(n^\omega)$, $\Pmat\Mmat=\Lmat\Umat$, where (1) $\Pmat$ is a permutation matrix, and (2) $\Lmat$ and $\Umat$ are, respectively, lower and upper triangular matrices.

Next, it follows a result from Moore and Nilsson~\cite{moore2001parallel} that any permutation matrix $\Pmat$ can be decomposed into six parallel row-elimination matrices. 
Thus, \Cref{lem:noancillaMatrix} follows from the technical lemma below on the low-depth parallel-row-elimination factorization of triangular matrices. 
\end{proof}

\subsection{CNOT Circuits with Triangular Topological Structures}

\begin{lemma}[Triangular Matrices]\label{lem:noancillaU}
\Cref{lem:noancillaMatrix} holds for any triangular matrix $\Mmat$.
\end{lemma}
\begin{proof}
In the proof, we will focus on the case where $\Mmat$ is an upper triangular matrix.
The proof can be naturally extended to lower triangular matrices.

Our algorithm uses a divide-and-conquer scheme. For illustration, we assume here that $n$ is large enough and $\frac n2,\frac12\log n,\frac n{\log n}$ are all integers to avoid rounding issue (the details will become clear later in the proof). We will explain how to handle general $n$ at the end of this subsection.
For simplicity, we first consider a randomized algorithm.
We will then derandomize it using \Cref{lem:derandom} below. 
The synthesis process is illustrated in \Cref{fig:noancillaAlgo}, which has five main steps.
\begin{equation*}
\Diamondblack \quad \quad \Diamondblack \quad \quad \Diamondblack \quad \quad \Diamondblack \quad \quad \Diamondblack
\end{equation*}
\paragraph*{Step 1 (Recursion).}
Denote $\Mmat$ as $\bmat{\Mmat_1&\Amat\\&\Mmat_2}$, where $\Amat$ is of size $\frac n2\times\frac n2$. 
After simultaneous recursive parallel row elimination on $\Mmat_1,\Mmat_2$, the diagonal blocks becomes the identity matrices, and the upper-right sub-matrix becomes $\Mmat_1^{-1}\Amat$.
This step can be computed in advance---independently from the recursion---in time $\widetilde O(n^\omega)$~\cite{AlmanW21,Strassen1969}.

\paragraph*{Step 2 (Find Random Layover).}
Divide $\Mmat_1^{-1}\Amat$ as $\bmat{\Amat_1^\top&\cdots&\Amat^\top_{\frac12\log n}}^\top$ where each $\Amat_i$ is in $\Fbb_2^{\frac n{\log n}\times\frac n2}$.
Our key observation here is: 
\begin{quote}
{\em If $\Amat_i$ is ``close to random'' (to be formally defined in \Cref{lem:derandom}), then it has a low-depth parallel row elimination.}
\end{quote}
To ensure the needed degree of randomness in our matrix structures, we use 
a classical idea of Valiant~\cite{ValiantHypercube,ValiantBrebner} 
for oblivious routing \cite{KaklamanisKT91}:
We generate a random $\Bmat_i$ of the same size for each $\Amat_i$ as its {\em layover}.
We also formulate $\Cmat_i=\Amat_i\oplus\Bmat_i$.
Note that, although they are correlated, $\Bmat_i$ and $\Cmat_i$ are individually random $\frac n{\log n}\times\frac n{\log n}$ matrices with entries from $\Fbb_2^{\frac12\log n}$.

\paragraph*{Step 3 (Generate Row-Traversal Sequence).}
We say matrix sequence $\Xmat_1,\ldots,\Xmat_T\in\GL(c,2)$ is \emph{row-traversal} if:
\begin{itemize}
\item $\Xmat_T=\Imat$, 
\item for any $p\in[c]$, every vector in $\Fbb_2^c\backslash\cbra{0^c}$ appears in the sequence $\Xmat_1[p,*],\ldots,\Xmat_T[p,*]$. 
\end{itemize}
Let $c=\frac12\log n$. 
We will apply \Cref{lem:traversal} below to obtain a row-traversal sequence with $T=O(\sqrt n)$.

\begin{figure*}[ht]
\centering
\scalebox{0.42}{\input{./pics/noancillaAlgo.tex}}
\caption{Main algorithm for in-place parallel Gaussian elimination.}\label{fig:noancillaAlgo}
\end{figure*}

\paragraph*{Step 4 (First Traverse).}
In this step, we add $\Cmat_i$ to $\Amat_i$ and then get $\Bmat_i$ for all $i$.

We view the bottom-right $\Imat_{n/2}$ as $n/\log n$ identity matrices 
of same size, and refer to them as $\Ymat_1,\ldots,\Ymat_{n/\log n}$.

Let $\Xmat_0=\Imat$.
For each time stamp $t\in[T]$, simultaneously all $\Ymat_j$'s go from $\Xmat_{t-1}$ to $\Xmat_t$ using original Gaussian elimination algorithm.
Then:
\begin{itemize}
\item for all $i\in\sbra{\frac12\log n}$, find ``large'' set $S_i\subseteq\sbra{\frac n{\log n}}\times\sbra{\frac n{\log n}}$ such that any $(j,k)\in S_i$ satisfies:
\begin{itemize}
\item $\Cmat_i[j,k]=\Xmat_t[i,*]$ (recall that entries of $\Cmat_i$ are from $\Fbb_2^{\frac12\log n}$),
\item $(j,k)$ was not selected in previous $S_i$ (i.e., $S_i$ in previous time stamps or repetitions),
\item for any other $(j',k')\in S_i$, we have $j'\neq j$ and $k'\neq k$;
\end{itemize}
\item for all $i$ and $(j,k)\in S_i$, add row-$i$ of $\Ymat_k$ to row-$j$ of $\Amat_i$ as one parallel row elimination;
\item repeat the two steps above until all $S_i=\emptyset$.
\end{itemize}
The detailed explanation for how to construct $S_i$ will be provided below immediately after Step 5.

\paragraph*{Step 5 (Second Traverse).}
Now in the upper-right part, all $\Amat_i$'s have reached the pre-decided layout $\Bmat_i$'s. 
In this step, we do another round of traverse similar to that of Step 4; 
the only difference is that we use $\Bmat_i$ when constructing $S_i$. 
Thus, we add $\Bmat_i$'s to $\Bmat_i$'s as in Step 4 in order to turn the 
upper-right square into the zero matrix.
\begin{equation*}
\Diamondblack \quad \quad \Diamondblack \quad \quad \Diamondblack \quad \quad \Diamondblack \quad \quad \Diamondblack
\end{equation*}

Now we explain the construction of $S_i$ in Step 4. 
For fixed $t$ and $i$, although $S_i$ is formulated iteratively in our description above (in favor of clarity of presentation), it is actually implemented in a single shot. 
We now provide its detail where the random $\Bmat_i$ plays an essential role.

When $\Bmat_i$ is random, with high probability, any vector in $\Fbb_2^{\frac12\log n}$ appears $O\pbra{\sqrt n/\log n}$ times in every row and column of $\Cmat_i$.
Then, we enumerate all $(j,k)\in\sbra{\frac n{\log n}}\times\sbra{\frac n{\log n}}$ such that $\Cmat_i[j,k]=\Xmat_t[i,*]$ and view them as the edges on a bipartite graph. 
Thus, any valid $S_i$ is a matching in this graph and our iterative construction is equivalent to a matching decomposition. 
Since any vertex has degree $O\pbra{\sqrt n/\log n}$, the bipartite graph can be factorized into $O\pbra{\sqrt n/\log n}$ matchings in near linear time \cite{ColeOS01,Alon03}.

Hence, Step 4 uses $O\pbra{\sqrt n/\log n}$ parallel row-elimination matrices for every time stamp. 
Similar analysis holds for Step 5.
Below, we will derandomize the choice of $\Bmat_i$ using \Cref{lem:derandom}.

Putting together, the maximum number of parallel row-elimination matrices, denoted as $d(n)$, can be obtained as $O(n/\log n)$ by the following recursion:
$$
d(n)\leq
\underbrace{d\pbra{\frac n2}}_\text{Step 1}+
\underbrace{2\cdot}_\text{Step 4,5}
\underbrace{O(\sqrt n)}_T\cdot
\pbra{
\underbrace{O\pbra{\log^3n}}_{\text{Move }\Ymat_j}+
\underbrace{O\pbra{\sqrt n/\log n}}_\text{Matchings}
}.
$$

Using \Cref{lem:derandom} and \Cref{lem:traversal} (which will be explained shortly), the running time, denoted as $T(n)$, can be obtained as $\widetilde O(n^\omega)$ by the following recursion:
\begin{equation*}
T(n)\leq
\underbrace{2\cdot T\pbra{\frac n2}+\widetilde O(n^\omega)}_\text{Step 1}+
\underbrace{\widetilde O\pbra{n^2}}_\text{Step 2}+
\underbrace{\widetilde O(n)}_\text{Step 3}
+\underbrace{2\cdot}_\text{Step 4,5}
\pbra{
\underbrace{\widetilde O\pbra{n^2}}_\text{Get $S_i$}+
\underbrace{\widetilde O\pbra{n^2}}_\text{Add $\Ymat_k$ to $\Amat_i$}+
\underbrace{O(\sqrt n)}_T\cdot
\underbrace{\widetilde O(1)}_{\text{Move }\Ymat_j}
}.
\tag*{\qedhere}
\end{equation*}
\end{proof}

Now we build the two essential components needed in the proof of \Cref{lem:noancillaU}. 

\Cref{lem:derandom} below addresses the crucial property that $\Bmat_i,\Cmat_i$ must have in order to make the matching decomposition and parallel row elimination efficient.
Note that the proof of existence in \Cref{lem:derandom} can be obtained by a direct application of Chernoff-Hoeffding inequality.
However, the resulting construction is hard to derandomize, particularly, in time $\widetilde O(n^2)$.
In the proof, we present an efficient deterministic construction. 

\begin{lemma}[Efficient Layover]\label{lem:derandom}
There is an $\widetilde O\pbra{n^2}$-time algorithm such that for sufficiently large $n$, given any $\frac n{\log n}\times\frac n{\log n}$ matrix $\Amat$ with entries from $\Fbb_2^{\frac12\log n}$, it outputs $\Bmat$ of the same format satisfying 
that for all $\bm v\in\Fbb_2^{\frac12\log n}$, $\bm v$ appears at most $\frac{\sqrt{en}}{\log n}$ times in any row or column of $\Bmat$ and of $\Amat\oplus\Bmat$.
\end{lemma}
\begin{proof}
Let $\Bmat$ be a random $\frac n{\log n}\times\frac n{\log n}$
matrix with i.i.d. entries uniformly selected from $\Fbb_2^{\frac12\log n}$.
Let $\Cmat=\Amat\oplus\Bmat$ and set $\epsilon=(2\log n)^{-1}$ with foresight.

To set up both the existence proof and deterministic construction, we view $\Bmat$ as a matrix where one obtains its entries by picking bit by bit uniformly at random.
In the following, we prove by induction on the number of determined bits $z=0,\ldots,\frac12\log n$, that no vector in $\Fbb_2^z$ appears as a prefix more than $\pbra{\frac12+\epsilon}^z\frac n{\log n}$ times in any row or column of $\Bmat,\Cmat$.

Assume that the first $k-1$ bits are determined. Now we randomly assign the $k$-th bit from $\bin$. 
For any $\bm u\in\Fbb_2^k,i\in\sbra{\frac n{\log n}}$, we consider the indicators for the following four 0/1 ``bad events'':
\begin{itemize}
\item $B_{\bm u}^\text{row-$i$}=1$ ($C_{\bm u}^\text{row-$i$}=1$) iff $\bm u$ appears as a prefix more than $\pbra{\frac12+\epsilon}^k\frac n{\log n}$ times in $\Bmat[i,*]$ ($\Cmat[i,*]$),
\item $B_{\bm u}^\text{col-$i$}=1$ ($C_{\bm u}^\text{col-$i$}=1$) iff $\bm u$ appears as a prefix more than $\pbra{\frac12+\epsilon}^k\frac n{\log n}$ times in $\Bmat[*,i]$ ($\Cmat[*,i]$).
\end{itemize}
Let $K=\pbra{\frac12+\epsilon}^{k-1}\frac n{\log n}$ and $\text{Bin}(K,1/2)$ be the $K$-round Bernoulli trial with probability $1/2$.
Then the expected number of bad events is:
\begin{align*}
\E\sbra{\sum_{\bm u,i}B_{\bm u}^\text{row-$i$}+B_{\bm u}^\text{col-$i$}+C_{\bm u}^\text{row-$i$}+C_{\bm u}^\text{col-$i$}}
\leq & \ 2^k\cdot\frac{4n}{\log n}\cdot
\Pr\sbra{\text{Bin}\pbra{K,\frac12}>\pbra{\frac12+\epsilon}K}\\
\leq &\ \frac{4n\sqrt n}{\log n}\cdot e^{-2K\epsilon^2}=o(1).
\tag{by Chernoff-Hoeffding inequality}
\end{align*}

Therefore, there exists an assignment of the $k$-th bit such that no $\bm u\in\Fbb_2^k$ appears as a prefix for more than $\pbra{\frac12+\epsilon}^k\frac n{\log n}$ times in any row or column of $\Bmat,\Cmat$ as stated.
This desired property follows from
$$
\pbra{\frac12+\epsilon}^{\frac12\log n}\frac n{\log n}\leq\frac{\sqrt{en}}{\log n}.
$$
We use $\bm \cdot\stackrel{\ell}{=}\bm \cdot$ to denote the relation that two vectors share the same first $\ell$ bits. 
Let the undetermined bit in entries of $\Bmat,\Cmat$ be $*$.

Now, for fixed $i,j$, we derandomize the assignment of the $k$-th bit of $\Bmat[i,j]$ by the method of conditional expectation. 
Let the first $k-1$ bits of $\Bmat[i,j]$ and $\Cmat[i,j]$ be $\bm p$ and $\bm q$ respectively. 
Define
$$
r_\clubsuit^{\spadesuit,\diamondsuit}=\#\cbra{y\in\sbra{\frac n{\log n}}\middle|\clubsuit[i,y]\stackrel k=\spadesuit\diamondsuit}
\quad\text{and}\quad
c_\clubsuit^{\spadesuit,\diamondsuit}=\#\cbra{y\in\sbra{\frac n{\log n}}\middle|\clubsuit[y,j]\stackrel k=\spadesuit\diamondsuit},
$$
where $(\clubsuit,\spadesuit)\in\cbra{(\Bmat,\bm p),(\Cmat,\bm q)},\diamondsuit\in\cbra{0,1,*}$.

Let $s=\pbra{\frac12+\epsilon}^k\frac n{\log n}$ and the $k$-th bit of $\Amat[i,j]$ be $b$.
Suppose we pick $b_{\Bmat}$ as the $k$-th bit of $\Bmat[i,j]$, and define $b_{\Cmat}=b_{\Bmat}\oplus b,\bar b_{\Bmat}=b_{\Bmat}\oplus1,$ and
$\bar b_{\Cmat}=b_{\Cmat}\oplus1$, then the expectation of bad events 
decreases by:
$$
\sum_{\kappa\in\cbra{r,c},(\clubsuit,\spadesuit)}
\Bigg[
\sum_{\diamondsuit\in\bin}
f\pbra{\kappa_\clubsuit^{\spadesuit,*},s-\kappa_\clubsuit^{\spadesuit,\diamondsuit}}
-
f\pbra{\kappa_\clubsuit^{\spadesuit,*}-1,\kappa_\clubsuit^{\spadesuit,b_{\clubsuit}}-1}-
f\pbra{\kappa_\clubsuit^{\spadesuit,*}-1,\kappa_\clubsuit^{\spadesuit,\bar b_{\clubsuit}}}\Bigg],
$$
where
$$
f(u,v)=\Pr\sbra{\text{Bin}\pbra{u,\frac12}>v}=2^{-u}\sum_{i>v}\binom ui.
$$
We choose $b_{\Bmat}\in\bin$ with the largest decrease, which must be non-negative.

To speed-up the selection, we pre-process and truncate $f(u,v)$ to the highest $\log^2n$ significant bits. 
Naturally, the truncation can introduce some errors, particularly, the best truncated $b_{\Bmat}$ can increase the expectation.
However, because the fluctuation is sufficiently small, i.e., bounded by $O\pbra{2^{-\log^2n}}$, the errors accumulate as an insignificant $o(1)$ term.
\end{proof}

\Cref{lem:traversal} below presents a simple construction of the row-traversal sequence. 
Although one can further improve the length of the sequence, the achieved order is asymptotically tight and sufficient for our purpose.
Thus, in favor of the simplicity of the presentation, we won't pursue the optimal parameter for the construction.

\begin{lemma}[Row-Traversal Sequence]\label{lem:traversal}
There is an $O\pbra{\poly(k)2^k}$-time algorithm to generate a row-traversal sequence of length $3\cdot 2^{k-1}-k+1$.
\end{lemma}
\begin{proof}
Let $\bm1=\bmat{1&\cdots&1}^\top$.
Let $\rank(\cdot)$ denote the function for computing the rank of matrices over $\Fbb_2^{k\times k}$.
Observe that for any $\bm v\in\Fbb_2^k$,
$$
\rank\pbra{\Imat\oplus\bm1\bm v^\top}=\begin{cases}
k-1&\bm v^\top\bm1=1,\\
k&\bm v^\top\bm1=0,
\end{cases}
$$
since for any $\bm x\neq\bm0$,
$$
\bm x^\top\pbra{\Imat\oplus\bm1\bm v^\top}\begin{cases}
=\bm0^\top&\pbra{\bm x=\bm v}\land\pbra{\bm v^\top\bm1=1},\\
\neq\bm0^\top&\text{otherwise}.
\end{cases}
$$

We will use the following algorithm.

\begin{algorithm}[H]
\caption{Row-traversal sequence}\label{alg:traversal}
\ForEach{$\bm v\in\Fbb_2^k$}{
    $\Tmat\gets\Imat\oplus\bm 1\bm v^\top$\;
	\If(\tcc*[f]{$\Tmat$ is not invertible}){$\bm v^\top\bm 1=1$}{
		let $i$ be an arbitrary index that $v_i=1$\;
		$\bm u\gets\Tmat[i,*]$ and change $\Tmat[i,*]$ to make $\Tmat$ invertible\;
		\If(\tcc*[f]{Compensate for replacing $\bm u$}){$\bm u\neq\bm0$}{
			\textbf{output} some $\Vmat\in\GL(k,2)$ that $\Vmat[i,*]=\bm u$
		}
	}
	\textbf{output} $\Tmat$
}
\textbf{output} $\Imat$
\end{algorithm}

Note that any row of $\Imat\oplus\bm1\bm v^\top$ traverses all vectors in $\Fbb_2^k$ as $\bm v$ goes through $\Fbb_2^k$.
Thus, \Cref{alg:traversal} generates a desired sequence.
\end{proof}

\paragraph*{Handling General $n$.} 
For $n$ being general that $\frac n2,\frac12\log n,\frac n{\log n}$ are not necessarily integers, we make several modifications to the algorithms in \Cref{lem:noancillaU} as illustrated in \Cref{fig:generaln}. Recall that we have no ancilla, thus we can not simply extend $n$ to some $n'\geq n$ satisfying those conditions and then claim \Cref{lem:noancillaU} holds for $n$ since it holds for $n'$. The details are as follows.

\begin{figure*}[ht]
\centering
\scalebox{0.42}{\input{./pics/generaln.tex}}
\caption{Main algorithm for general $n$.}\label{fig:generaln}
\end{figure*}

In the \textit{Step 1 (Recursion)}, denote $\Mmat$ as $\bmat{\Mmat_1&\Amat\\&\Mmat_2}$ where $\Mmat_1$ is of size $\ceilbra{ \frac{n}{2}} \times \ceilbra{ \frac{n}{2}}$, $\Mmat_2$ is of size $\floorbra{ \frac{n}{2} } \times \floorbra{ \frac{n}{2} }$, and $\Amat$ is of size $\ceilbra{ \frac{n}{2}} \times \floorbra{ \frac{n}{2} }$. As usual, after the simultaneous recursion on $\Mmat_1,\Mmat_2$, the upper-right sub-matrix is again $\Mmat_1^{-1}\Amat$.

Note that now $\Mmat_1^{-1}\Amat$ is not necessarily a square matrix, and $\frac{1}{2}\log n,\frac{n}{\log n}$ may not be integers. 
Therefore, let
$$
g=\ceilbra{ \frac{1}{2}\log n },\quad
g'=\floorbra{ \frac{n}{2} } \mod g,
\quad\text{and}\quad
h=\floorbra{ \frac{\ceilbra{ \frac{n}{2} }}{g} },\quad
h'=\ceilbra{ \frac{n}{2} } \mod g. 
$$ 
In \textit{Step 2 (Find Random Layover)}, We divide $\Mmat_1^{-1}\Amat$ as $\bmat{\Amat_1^\top&\cdots&\Amat^\top_{g+1}}^\top$, where each $\Amat_i, i\in[g]$ has $h$ rows and $\Amat_{g+1}$ has $h'$ rows.

As in the first sub-figure of \Cref{fig:generaln}, $\Mmat_1^{-1}\Amat$ can be viewed as one regular part (the part under slash, i.e., the upper-left $(\ceilbra{\frac n2}-h')\times(\floorbra{\frac n2}-g')$ sub-matrix) and one corner by it. 
We now eliminate the regular part using our original algorithm. 
Specifically, in \textit{Step 3 (Generate Row-Traversal Sequence)} we set $c=g$, and view the bottom-right $\Imat_{\floorbra{n/2}}$ as many $\Imat_g$'s along with one $\Imat_{g'}$ which we do not use for now. 
Then we generate the layover $\Bmat,\Cmat$'s for the regular part, and perform \textit{Step 4 (First Traverse)} and \textit{Step 5 (Second Traverse)} to eliminate it to zero. 
Here we note that the derandomization (\Cref{lem:derandom}), originally proved for square matrices and special $n$, can be easily adapted here.

To eliminate the corner part, we simply use $3$ shifts of regular parts as illustrated in \Cref{fig:generaln}. More precise, we define
\begin{itemize}
\item \texttt{DA-1}/\texttt{DA-2}/\texttt{DA-3}/\texttt{DA-4} to, respectively, view the upper-left/bottom-left/upper-right/bottom-right $(\ceilbra{\frac n2}-h')\times(\floorbra{\frac n2}-g')$ sub-matrix of $\Mmat_1^{-1}\Amat$ as the regular part.
\item \texttt{DI-1} to view the bottom-right $\Imat_{\floorbra{n/2}}$ as many $\Imat_g$'s along with one $\Imat_{g'}$.
\item \texttt{DI-2} to view the bottom-right $\Imat_{\floorbra{n/2}}$ as one $\Imat_{g'}$ along with many $\Imat_g$'s.
\end{itemize}
Then in \Cref{fig:generaln},
\begin{itemize}
\item sub-figure 1 to 2 is, what we detailedly described above, the combination of \texttt{DA-1} and \texttt{DI-1};
\item sub-figure 2 to 3 is the combination of \texttt{DA-2} and \texttt{DI-1};
\item sub-figure 3 to 4 is the combination of \texttt{DA-3} and \texttt{DI-2};
\item sub-figure 4 to 5 is the combination of \texttt{DA-4} and \texttt{DI-2}.
\end{itemize}

To analyze the performance of the algorithm for general $n$, we rewrite the recursion for depth (i.e., $d(n)$) and running time (i.e., $T(n)$):
\begin{align*}
&d(n)\leq
\underbrace{\max\cbra{d\pbra{\ceilbra{ \frac{n}{2} }},d\pbra{\floorbra{\frac n2}}}}_\text{recursion}+
\underbrace{4\cdot}_\text{4 shifts}
\underbrace{O\pbra{\frac n{\log n}}}_\text{eliminate the regular part}
=O\pbra{\frac n{\log n}}\\
&\text{and}\quad T(n)\le T\pbra{\ceilbra{\frac n2}}+T\pbra{\floorbra{\frac n2}}+\widetilde O\pbra{n^\omega}=\widetilde O\pbra{n^\omega}.
\end{align*}

\paragraph*{Extension to General Domain.}
The techniques above for proving \Cref{lem:noancillaMatrix} can be extended to general $\Mmat\in\GL(n,q)$ to establish \Cref{thm:FfieldParallel}.

\begin{proof}[Proof of \Cref{thm:FfieldParallel}]
The proof is almost identical except that each $\log$ is now replaced with $\log_q$.
We only note that the change to \Cref{lem:traversal}: For $\GL(n,q)$, the length of the row-traversal sequence becomes $(q+1)q^{k-1}-k+1$.
\end{proof}

\subsection{Topological Structures in CNOT-Circuit Synthesis}
\label{sec:TopologicalStructures}

A fundamental problem in parallel CNOT-circuit synthesis~---~with or without ancillae~---~is to characterize the impact of circuits' topological structures on the size-depth trade-off.
In the asymptotic space-depth trade-off of the previous section, CNOT circuits are compressed essentially as an invertible matrix in $\GL(n,2)$.
There, we have shown that any $n$-qubit CNOT circuit can be transformed into an equivalent one with $O(n/\log n)$ depth without ancillae. 
To exploit topological structures in circuit optimization, the circuit layout itself becomes a part of the input to synthesis algorithms.

In the following, we use a basic family of CNOT circuits to illustrate that the topological details of CNOT circuits can be effectively used.
This family of the CNOT circuits has {\em tree structures}, and in this section, we will consider two types of tree structures. Consider a {\em tree} $\Tcal$ with $n$ nodes, in which each node $i$ corresponds to a qubit $|x_i\rangle$ in the circuit, and each edge $i\rightarrow j$ corresponds to a CNOT gate with control qubit $\ket{x_i}$ and target qubit $\ket{x_j}$, denoted as $\Rsfmat(i,j)$. The first tree structure is out-arborescence where the direction of each edge in the tree is from the node's parent to the node. The order of the CNOT gate corresponds to the breadth-first search from the root. The second tree structure is the opposite of the first one. It is in-arborescence where the direction of each edge in the tree is from the node to its parent. The order of the CNOT gate corresponds to the reverse order of breadth-first search from the root. See \Cref{fig:outtreeexample} and \Cref{fig:intreeexample} for the illustration examples. 

\begin{figure*}[ht]
\centering
\scalebox{1}{\input{./pics/outtreeexample.tex}}
\caption{An example of CNOT out-arborescence and its corresponding circuit.}\label{fig:outtreeexample}
\end{figure*}

\begin{figure*}[ht]
\centering
\scalebox{1}{\input{./pics/intreeexample.tex}}
\caption{An example of CNOT in-arborescence and its corresponding circuit.}\label{fig:intreeexample}
\end{figure*}

We call such tree structure $\Tcal$ {\em a CNOT out-arborescence} (resp. {\em a CNOT in-arborescence}) and the corresponding CNOT circuit $\Mmat_\Tcal$ {\em an out-arborescence-based circuit} (resp. {\em an in-arborescence-based circuit}).
It is oblivious that not all CNOT circuits are arborescence-based circuits. For two equivalent CNOT circuits, it is also possible that one of them is arborescence-based while the other one is not. 

The following theorem gives an efficient method for generating an equivalent $O(\log ^2 n)$-depth CNOT circuit for any $n$-qubit arborescence-based CNOT circuit $\Mmat_\Tcal$.

\begin{theorem}[Parallel Synthesis of CNOT Arborescences]\label{thm:ParCNOTTree}
For any CNOT arborescence $\Tcal$ with $n$ nodes (either in-arborescence or out-arborescence), the $n$-qubit CNOT circuit $\Mmat_\Tcal$ can be efficiently parallelized to $O(\log^2 n)$ depth without ancillae.
\end{theorem}
\Cref{thm:ParCNOTTree} can be obtained by applying Miller and Reif's parallel-tree-contraction technique ~\cite{miller1985parallel,miller1989parallel,miller1991parallel}. 
First, we consider three specific arborescence-based CNOT circuits which can be parallelized to $O(\log n)$ depth. See \Cref{prop:path}, \Cref{prop:inarborescence} and \Cref{prop:outarborescence}.
We then use these constructions as well as Miller and Reif's parallel-tree-contraction technique to obtain the proof of \Cref{thm:ParCNOTTree}.

\begin{proposition}\label{prop:path}
For the circuit $\Mmat$ based on the CNOT path in \Cref{fig:path}, we have parallel row-elimination matrices $\Rmat_1,\ldots, \Rmat_d$ such that $\Rmat_d\cdots \Rmat_{1} = \Mmat$ and $d=O(\log n)$.
\end{proposition}

\begin{figure*}[ht]
\centering
\scalebox{1}{\input{./pics/path.tex}}
\caption{A CNOT path.}\label{fig:path}
\end{figure*}

\begin{proof}
The corresponding matrix of circuit based on CNOT tree in \Cref{fig:path} is 
$$
\Mmat =\bmat{
1&&\\ 
1& 1& & \\
\vdots&\vdots&\ddots&\\
1 & 1&\cdots&1}.
$$
We focus on the case when $n$ is a perfect power of $2$.
In this case, let $d=(2\log n-1)$.
We use the following process to construct parallel row-elimination matrices:
\begin{itemize}
\item In layer $j\in[\log n]$, perform $\Rsfmat\pbra{2^j k + 2^{j - 1}, 2^{j}k + 2^{j}}$ for $0\leq k < n/2^j$.
\item In layer $(2\log n - j)$ with $j\in[\log n - 1]$, perform $\Rsfmat\pbra{2^j k , 2^{j}k + 2^{j-1}}$ for $0 < k < n/2^j$.
\end{itemize}
This construction can be naturally generalized to any integer $n > 0$.
\end{proof}

\begin{proposition}\label{prop:inarborescence}
For the circuit $\Mmat$ based on the CNOT in-arborescence $\Tcal$ in \Cref{fig:inarborescence}, we have parallel row-elimination matrices $\Rmat_1,\ldots, \Rmat_d$ such that $\Rmat_d\cdots \Rmat_{1} = \Mmat$ and $d=O(\log n)$.
\end{proposition}

\begin{figure*}[ht]
\centering
\scalebox{1}{\input{./pics/intree.tex}}
\caption{A special CNOT in-arborescence.}\label{fig:inarborescence}
\end{figure*}

\begin{proof}
The corresponding matrix of circuit based on CNOT in-arborescence in \Cref{fig:inarborescence} is 
$$
\Mmat =\bmat{
1&1&1&\cdots&1\\ 
0&1&0&\cdots&0\\
\vdots&\vdots&\vdots&\ddots&\vdots\\
0 & 0&0&\cdots&1}.
$$
We focus on the case when $n-1$ is a perfect power of $2$.
In this case, let $d=2\log (n-1)+1$.
We use the following process to construct parallel row-elimination matrices:
\begin{itemize}
\item In layer $j\in[\log (n-1)]$, perform $\Rsfmat\pbra{2^j k + 2^{j - 1}+1, 2^{j}k + 2^{j}+1}$ for $0\leq k < (n-1)/2^j$.
\item In layer $\log (n-1)+1$, perform $\Rsfmat\pbra{n, 1}$.
\item In layer $(2\log(n-1)+2 - j)$ with $j\in[\log (n - 1)]$, perform $\Rsfmat\pbra{2^j k + 2^{j - 1}+1, 2^{j}k + 2^{j}+1}$ for $0\leq k < (n-1)/2^j$.
\end{itemize}
This construction can be naturally generalized to any integer $n > 0$.
\end{proof}

\begin{proposition}\label{prop:outarborescence}
For the circuit $\Mmat$ based on the CNOT out-arborescence $\Tcal$ in \Cref{fig:outtree}, we have parallel row-elimination matrices $\Rmat_1,\ldots, \Rmat_d$ such that $\Rmat_d\cdots \Rmat_{1} = \Mmat$ and $d=O(\log n)$.
\end{proposition}

\begin{figure*}[ht]
\centering
\scalebox{1}{\input{./pics/outtree.tex}}
\caption{A special CNOT out-arborescence.}\label{fig:outtree}
\end{figure*}

\begin{proof}
The corresponding matrix of circuit based on CNOT out-arborescence in \Cref{fig:outtree} is 
$$
\Mmat =\bmat{
1&0&0&\cdots&0\\ 
1&1&0&\cdots&0\\
1&0&1&\cdots&0\\
\vdots&\vdots&\vdots&\ddots&\vdots\\
1 & 0&0&\cdots&1}.
$$
We focus on the case when $n-1$ is a perfect power of $2$.
In this case, let $d=2\log (n-1)+1$.
We use the following process to construct parallel row-elimination matrices:
\begin{itemize}
\item In layer $j\in[\log (n-1)]$, perform $\Rsfmat\pbra{2^j k + 2, 2^{j}k + 2^{j-1}+2}$ for $0\leq k < (n-1)/2^j$.
\item In layer $\log (n-1)+1$, perform $\Rsfmat\pbra{1, 2}$.
\item In layer $(2\log(n-1)+2 - j)$ with $j\in[\log (n - 1)]$, perform $\Rsfmat\pbra{2^j k + 2, 2^{j}k + 2^{j-1}+2}$ for $0\leq k < (n-1)/2^j$.
\end{itemize}
This construction can be naturally generalized to any integer $n > 0$.
\end{proof}

We now go back to prove \Cref{thm:ParCNOTTree}.
\begin{proof}[Proof of \Cref{thm:ParCNOTTree}]

We will focus on the result that for any CNOT in-arborescence $\Tcal$ with $n$ nodes, the $n$-qubit CNOT circuit $\Mmat_\Tcal$ can be efficiently parallelized to $O(\log^2 n)$ depth without ancillae. By this result, if we reverse the order of all CNOT gates and also change the control direction in each CNOT gate, we can parallelize any out-arborescence-based circuit within  $O(\log^2 n)$ depth without ancillae. We can also use very similar method for the in-arborescence case and the help of \Cref{prop:outarborescence} to obtain the same result. 

We use the following classic result of Miller and Reif~---~known as {\em parallel tree contraction}~---~in parallel algorithm design \cite{miller1985parallel,miller1989parallel,miller1991parallel}.
Miller and Reif introduced two parallel abstract operations on trees:
\begin{itemize}
\item {\sc Rake} - removing all leaves from the tree.
\item {\sc Compress} - replacing every maximal tree path by a tree path with one-half the length. Here, a tree path is a chain $v_1, v_2,\ldots, v_k$ where $v_{i-1}$ is the only child of $v_i$, for any $2\leq i\leq k$.
\end{itemize}
They proved that $O(\log n)$ rounds of ``first {\sc Rake} then {\sc Compress}'' reduce any rooted tree to its root.

Here, we will cast this rake-and-compress technique for our CNOT trees. We use the following two operations:
\begin{itemize}
\item {\sc CNOT-Rake} - removing all leaves from our in-arborescence. 
\item {\sc CNOT-Compress} - replacing every maximal path by an edge.
\end{itemize}
We firstly explain in details about these two operations. 

For {\sc CNOT-Rake} operation, suppose in-arborescence before the operation is $\Tcal$ and the in-arborescence after the operation is $\Tcal'$. Suppose in $\Tcal$, there are $k$ leaves $i_1, i_2,\ldots, i_k$ and their corresponding parents are $j_1, j_2,\ldots, j_k$. It is oblivious that 
$i_1, i_2,\ldots, i_k$ are different, and any $j_s$ is not a leaf. Note that it is possible that $j_s=j_t$ for some $s\neq t$. By the definition of in-arborescence-based circuit, we have $\Mmat_{\Tcal}= \Mmat_{\Tcal'}\cdot \prod_{s=1}^k \Rsfmat(i_s, j_s)$. We partition $(i_1,j_1), (i_2,j_2),\ldots, (i_k,j_k)$ into several groups where each group has the same $j$. Since each group is independent with other groups, then by \Cref{prop:inarborescence} the circuit $\prod_{s=1}^k \Rsfmat(i_s, j_s)$ can be parallelized within $O(\log n)$ depth. 

\begin{figure*}[ht]
\centering
\scalebox{1}{\input{./pics/cnotcompress.tex}}
\caption{Illustration of the {\sc CNOT-Compress} operation.}\label{fig:cnotcompress}
\end{figure*}

For the {\sc CNOT-Compress} operation, at first glance, it is the same as the case in \Cref{prop:path}. But note that we cannot directly apply the proposition since the start point of the maximal path might not be a leaf. See \Cref{fig:cnotcompress} for the illustration of the {\sc CNOT-Compress} operation. The first step is to add qubit $x_2,\ldots, x_{k-1}$ to qubit $x_k$. By \Cref{prop:inarborescence}, this can be done in $O(\log n)$ depth. Then, in the viewpoint of qubit $x_k$, the in-arborescence $\Tcal$ becomes $\Tcal'$. But remember we do not handle qubit $x_2,\ldots, x_{k-1}$. After the in-arborescence $\Tcal'$ is computed, we need to run the second step to compute the path $\Tcal''$. By \Cref{prop:path}, this can be done in $O(\log n)$ depth. More precisely, for the case in \Cref{fig:cnotcompress}, we have $\Mmat_{\Tcal}= \Mmat_{\Tcal''}\cdot \Mmat_{\Tcal'}\cdot \prod_{s=2}^{k-1} \Rsfmat(s, k)$. Since we only consider the maximal path, any node can be involved in at most one path. So for all maximal paths, both the first step and the second step can be done simultaneously. This means we can do one {\sc CNOT-Compress} operation within $O(\log n)$ depth. 

Finally, Miller and Reif's result actually shows that $O(\log n)$ rounds of ``first {\sc CNOT-Rake} then {\sc CNOT-Compress}'' reduce any in-arborescence to its root. Note that for the  {\sc CNOT-Compress} operation, instead of replacing every maximal tree path by a tree path with one-half the length, we actually replace every maximal tree path by an edge.   
\end{proof}

The following is an immediate corollary of \Cref{thm:ParCNOTTree}. Given $k$ CNOT arborescences $\Tcal_1,\ldots, \Tcal_k$, the product of $k$ CNOT arborescences is defined as the circuit $\Mmat = \Mmat_{\Tcal_k}\cdots \Mmat_{\Tcal_2}\Mmat_{\Tcal_1}$.
\begin{corollary}
If an $n$-qubit CNOT circuit can be expressed as the product of $k$ CNOT arborescences, then it can be efficiently transformed into an equivalent CNOT circuit of $O(k\log^2n)$ depth without ancillae.
\end{corollary}

\section{CNOT-Circuit Synthesis: Optimal Ancillae-Depth Trade-Off}\label{sec:withancilla}

In this section, we prove \Cref{thm:upperbound} for the case when $m=\Omega(n)$ (i.e., $m=sn,s=\Omega(1)$). 
Note that for any $s=\Omega(n/\log^2n)$, the bound in \Cref{thm:upperbound} is always $O(\log n)$.
Thus it suffices to consider $s = O(n/\log^2n)$. We assume $n$ is an even power of $2$ so $\sqrt{n}$ and $ \frac{1}{2}\log n$ are integers. For general $n$,
we choose the minimum integer $n'\geq n$ which is an even power of $2$, and embed our $n$-qubit circuit into an $n'$-qubit circuit by adding ancillae. Note that $n'\leq 4 n$ so it does not change our asymptotic results.

We restate \Cref{thm:upperbound} in this case as the following \Cref{thm:CNOTancilla}.

\begin{theorem}[Parallelizing CNOT Circuits]\label{thm:CNOTancilla}
For any $s$ in range $1\leq s\leq O(n/\log^2n)$, any $n$-qubit CNOT circuit can be synthesized into an equivalent circuit of $O\pbra{\frac{n}{s\log n}}$ depth with $(3s+1)n$ ancillae.
Moreover, there is an $\widetilde O(n^{\omega})$-time algorithm to for achieving this parallelization.
\end{theorem}

We use a standard technique in reversible computation to simplify the problem. 
Given arbitrary $\Mmat\in \GL(n,2)$, \Cref{thm:CNOTancilla} aims to construct a CNOT circuit for $\Mmat$ with ancillae, by first constructing two $2n$-qubit $3sn$-ancilla CNOT circuits $\Ccal_1,\Ccal_2$ for $\bmat{\Imat&\\\Mmat&\Imat},\bmat{\Imat&\\\Mmat^{-1}&\Imat}$ respectively.
In other words, for any $\bm x\in \Fbb_{2}^{n}$, $\bm j\in \Fbb_{2}^{n}$, 
\begin{align*}
&\Ccal_1\ket{\bm x}\ket{\bm j}\ket{0}^{\otimes 3sn}=\ket{\bm x}\ket{\bm j\oplus\Mmat\bm x}\ket{0}^{\otimes 3sn},\\
&\Ccal_2\ket{\bm x}\ket{\bm j}\ket{0}^{\otimes 3sn}=\ket{\bm x}\ket{\bm j\oplus \Mmat^{-1}\bm x}\ket{0}^{\otimes 3sn}.
\end{align*}

Starting with $\ket{\bm x}\ket{0}^{\otimes n}\ket{0}^{\otimes 3sn}$, we first apply $\Ccal_1$.
Then, before applying $\Ccal_2$, we permute the first and second $n$ qubits, which can be done in depth $6$ by~\cite{moore2001parallel}.
Call this permutation $\Scal_{n}$.
We get:
\begin{align*}
\ket{\bm x}\ket{0}^{\otimes n}\ket{0}^{\otimes 3sn}&\xrightarrow{\Ccal_1}\ket{\bm x}\ket{\Mmat\bm x}\ket{0}^{\otimes 3sn}\\
& \xrightarrow{\Scal_n}\ket{\Mmat\bm x}\ket{\bm x}\ket{0}^{\otimes 3sn}
\\
&\xrightarrow{\Ccal_2}\ket{\Mmat\bm x}\ket{\bm x\oplus \Mmat^{-1}\Mmat\bm x}\ket{0}^{\otimes 3sn}
= \ket{\Mmat\bm x}\ket{0}\ket{0}^{\otimes 3sn}.
\end{align*}

Based on the observation above together with the equivalence between CNOT circuits and invertible Boolean matrices, to prove \Cref{thm:CNOTancilla} it suffices to construct circuit for $\bmat{\Imat&\\\Mmat&\Imat}$ as in \Cref{lem:M}.

By \Cref{lem:M}, the time complexity to construct $\Ccal_1$ is $\widetilde{O}(n^2)$. 
To construct $ \Ccal_2$, we first need $\widetilde{O}\pbra{n^{\omega}}$ time to compute $\Mmat^{-1}$ by~\cite{Strassen1969,Gall14a}. 
The circuit for $\Scal_n$~\cite{moore2001parallel} is simply using swap gates on the $i$-th and $(i+n)$-th qubits for $i=1,\ldots,n$ which can be constructed in $O(n)$ time.
Thus, the overall time complexity is $\widetilde{O}\pbra{n^{\omega}}$ for 
the construction in \Cref{thm:CNOTancilla}.

\begin{lemma}\label{lem:M}
There is an $\widetilde O(n^2)$-time algorithm such that given $1\leq s \leq O(n/\log^2n)$ and $\Mmat\in \GL(n,2)$, it outputs parallel row-elimination matrices $\Rmat_1,\ldots,\Rmat_d$ where $d = O\pbra{\frac{n}{s\log n}}$, such that
$$
\Rmat_d\cdots\Rmat_1 = \begin{bmatrix}
\Imat & & \\
\Mmat & \Imat&\bm*\\
& &\Imat_{3sn}\\
\end{bmatrix}.
$$
\end{lemma}

At the high level, our approach to prove \Cref{lem:M} is the following: 
We partition the columns of $\Mmat$ into blocks of $s\log^2n$ columns.
We show that with the help of $3sn$ additional rows from ancillae, any such $s\log^2 n$ columns of $\Mmat$ can be constructed using $O(\log n)$ parallel Gaussian-elimination matrices.
We then sequentially assemble these blocks to synthesize $\Mmat$.

We first establish a few technical lemmas. We begin our proof with the base case when $s=1$ (to be restated in \Cref{lem:GenYCol}).
Our proof of \Cref{lem:GenYCol} calls on \Cref{lem:GenSparse} below, which provides a subroutine for constructing a slim sparse matrix. 

\begin{lemma}\label{lem:GenSparse} 
There is an $\widetilde O(\sqrt n)$-time algorithm that,
given $t\leq \frac{\sqrt{n}}{2}\log n$ and $\Ymat\in \Fbb_{2}^{\sqrt{n}\times \frac{1}{2}\log n}$ with  $t$ one's,
 outputs $\Rmat_1, \ldots, \Rmat_{d}$ where $d=O(\log n)$, such that 
$$
\Rmat_d\cdots\Rmat_1
=
\begin{bmatrix}
\Imat&&\\
\Ymat&\Imat&\bm*\\
&&\Imat_{t}
\end{bmatrix}.
$$
\end{lemma}
\begin{proof}
For each $j\in\sbra{\frac12\log n}$, we use $\bm e_j$ to denote the vector whose entries are $0$ except for the $j$-th entry.
Let $t_j=\#\{i|\Ymat[i,j]=1\}$. 
By the assumption of the lemma, we have $t=\sum_j t_j$.

We first focus on $j=1$ and assume without loss of generality $t_1 > 0$.
We now describe how to make $t_1$ copies of $\bm e_1$'s on any $t_1$ rows by $O(\log t_1)$ parallel row eliminations.
The construction can be naturally extended to other $j$.
In the construction, we will use the upper-left $\Imat_{\frac12\log n}$.

Let the $t_1$ rows be $a_1, \ldots, a_{t_1}$.
We add the first row (which is $\bm e_1$) to the $a_1$-th row.
We then double the number of $\bm e_1$ by adding the first and $a_1$-th rows to rows $a_2$ and $a_3$ simultaneously.
We will continue this doubling process with a proper last step till the number reaches $t_1$.

Since $\sum_i t_i=t\leq\frac12\sqrt{n}\log n$, we can make $t_i$ copies of $\bm e_i$ independently for all $i$ on the last $t$ rows with $O(\log n)$ parallel row-elimination matrices.

Then, we construct $\Ymat$ on the middle $\sqrt{n}$ rows $b_1,\ldots,b_{\sqrt n}$.
For $i\in\sbra{\sqrt n}$, we add $\bm e_j$'s to the $b_i$-th row if $\Ymat[i,j]=1$, which needs at most $\frac12\log n$ row additions. 
Since there are sufficient copies of $\bm e_j$'s, all $i\in\sbra{\sqrt n}$ are considered simultaneously.

At last, we restore the last $t$ rows by reversing the copy process.
\end{proof}

In the following lemma, corresponding to the case $s=1$, we use $O(\log n)$ parallel row-elimination matrices to construct any given $\log^2 n$ columns.

\begin{lemma}\label{lem:GenYCol}
There is an $\widetilde O(n)$-time algorithm that,
given $\Ymat\in\Fbb_2^{n\times \log^2 n}$, outputs parallel row-elimination matrices $\Rmat_1,\ldots,\Rmat_d$ where $d = O(\log n)$, such that 
$$
\Rmat_d\cdots\Rmat_1
=\begin{bmatrix}
\Imat&&\\
\Ymat&\Imat&\bm*\\
&&\Imat_{2n}
\end{bmatrix}.
$$
\end{lemma}
\begin{proof}
Rows of $\Ymat$ can be seen as a set of Boolean vectors of length $\log^2 n$. In the algorithm, we first synthesize an additive base for these vectors.
We then add them together to obtain $\Ymat$.
The main process is depicted in \Cref{fig:log2n}.
It applies the following four steps.

\begin{figure}[ht]
\centering
\scalebox{0.44}{\input{./pics/ancillaprogress.tex}}
\caption{The process (step 1-3) to construct $\Ymat$, (only the first $\log^2n$ columns) are drawn.}\label{fig:log2n}
\end{figure}

\paragraph*{Step 1 (Construct $\Pmat_i$'s).}
There are $\sqrt{n}$ distinct vectors in $\Fbb_2^{\frac{1}{2}\log n}$. 
With arbitrary ordering, we write $\bm y^\ell, \ell\in[\sqrt{n}]$ as the $\ell$-th vector in
$\Fbb_2^{\frac{1}{2}\log n}$. Let $\Pmat\in\Fbb_2^{\sqrt{n}\times \frac{1}{2}\log n}$ be the matrix where $\Pmat[\ell,*]=\bm y^\ell$. 
In the following, we will construct several copies of $\Pmat$ in parallel, as shown in \Cref{fig:log2n}.

Specifically, for $0\leq j \leq 2\log n - 1$, we apply \Cref{lem:GenSparse} with $O(\log n)$ parallel row-elimination matrices and ancillary row-$\pbra{\log^2 n+2n+\frac{j}{2}\sqrt{n}\log n+1}$ to
row-$\pbra{\log^2 n+2n+\frac{j+1}{2}\sqrt{n}\log n}$,
to construct $\Pmat$ in the intersection of
$\pbra{\log^2 n+n+j\sqrt{n}+1}$-th row to
$\pbra{\log^2 n+n+(j+1)\sqrt{n}}$-th row and
$\pbra{\frac j2\log n + 1}$-th column to
$\pbra{\frac{j + 1}2\log n}$-th column. 
For simplicity, we write the constructed $\Pmat$'s as $\Pmat_1,\Pmat_2,\ldots,$ $\Pmat_{2\log n}$. Notice that $\Pmat_i$ can be constructed simultaneously, thus the total cost is still $O(\log n)$ parallel row-elimination matrices. Here for simplicity we use $n$ ancillary rows to store those $\Pmat_i$, though this number can be slightly reduced by finer analysis.

\paragraph*{Step 2 (Copy rows in $\Pmat_i$'s).}
We divide $\Ymat$ as $\bmat{\Ymat_1&\cdots&\bm Y_{2\log n}}$ where $\Ymat_k\in \Fbb_2^{n\times \frac{1}{2}\log n}$ for all $k$. 
Suppose $\bm y^\ell$ appears $s_{k,\ell}$ times in $\Ymat_k$, i.e.,
$$
s_{k,\ell}=\#\cbra{u\in [n]\middle| \Ymat_k[u,*]=\bm y^\ell}.
$$
In principle, similar to \Cref{lem:GenSparse}, if we make $s_{k,\ell}$ copies of $\bm y^\ell$ from column-$\pbra{\frac{k-1}{2}\log n+1}$ to column-$\pbra{\frac{k}{2}\log n}$, then we can construct $\Ymat$ with $O(\log n)$ parallel row-elimination matrices. 
But a direct application of this approach will require $2n\log n$ ancillary rows. 
To achieve the optimality in ancillae-depth trade-off, we will use the following modification, given by the next two steps, that requires fewer initial copies.

We make $\ceilbra{\frac{s_{k,\ell}}{2\log n}}$ copies of $\bm y^\ell$ in column-$\pbra{\frac{k-1}{2}\log n+1}$ to column-$\pbra{\frac{k}{2}\log n}$ in parallel for all $k,\ell$ using $O(\log n)$ parallel row-elimination matrices.
Note that in this step, the rows in the original $\Pmat_i$ is regarded as the first copy of corresponding $\bm y^\ell$. 
This way, the ancillary rows being used during this step are at most $\frac{2n\log n}{2\log n}=n$.
Because $s_{k,\ell}\leq n$ for all $k,\ell$, Step 2 can be realized by $O(\log n)$ parallel row-elimination matrices.

\paragraph*{Step 3 (Construct $\Ymat$).}
For $i\in[n]$, let $u_i$ denote the $\pbra{\log^2 n+i}$-th row. 
For $j\in[2n]$, let $w_j$ denote the $\pbra{\log^2 n+n+j}$-th row, which corresponds to a copy of certain $\bm y^{\ell}$.

Our next step is to add proper $w_j$ to $u_i$ to form $\Ymat$. 
We say $w_j$ is used for $t$ times if $w_j$ is planned to be added to $t$ different $u_i$'s.
By Step 2, we have generated enough copies of the basic vectors to ensure that each $w_j$ will be used at most $2\log n$ times. 
Notice that $u_i$ is supposed to be the sum of $2\log n$ corresponding $w_j$s. We use the matching of the following bipartite graph $G=(V_1,V_2,E)$ to guide the process.
In $G$, $V_1=\{u_i\}_{i\in[n]}$, $V_2=\{w_j\}_{j\in [2n]}$, and $E=\cbra{(u_i,w_j)\middle|\text{$w_j$ will be added to $u_i$}}$.
Instrumental to our construction, all vertices in $G$ have degree at most $2\log n$. 
Thus, it can be factorized into at most $2\log n$ matchings in nearly linear time \cite{ColeOS01,Alon03}.
Hence, $\Ymat$ can be constructed from these copies of $\Pmat$ with $O(\log n)$ parallel row-elimination matrices.

The steps above are illustrated in \Cref{fig:log2n}.

\paragraph*{Step 4 (Restore).}
In this step, we erase the copies of $\bm y^\ell$'s and $\Pmat_i$'s 
 by the inverse of Step 2 and Step 1. The complexity is the same as before.

To sum up, the overall time complexity is
$$
\underbrace{\widetilde O\pbra{\sqrt{n}}}_\text{Step 1}+ \underbrace{\widetilde{O}(n)}_{\text{Step 2}}+ \underbrace{\widetilde{O}(n)}_{\text{Step 3}}+
\underbrace{\widetilde O\pbra{\sqrt{n}}+\widetilde{O}(n)}_{\text{Step 4}},
$$
and the total number of parallel row-elimination matrices is
\begin{equation*}
\underbrace{ O\pbra{\log n}}_\text{Step 1}+ \underbrace{O\pbra{\log n}}_{\text{Step 2}}+ \underbrace{ O\pbra{\log n}}_{\text{Step 3}}+
\underbrace{ O\pbra{\log n}}_{\text{Step 4}}.
\tag*{\qedhere}
\end{equation*}
\end{proof}

Our construction above can be naturally extended to cases when $s > 1$.

\begin{corollary}\label{cor:SParallel}
There is an $\widetilde{O}(sn)$-time algorithm that, given $1\leq s \leq O(n/\log^2n)$ and $\Ymat\in \Fbb_{2}^{n \times s\log^2 n}$, outputs parallel row-elimination matrices $\Rmat_1,\ldots,\Rmat_d$ where $d = O(\log n)$, such that
$$
\Rmat_d\cdots\Rmat_1 = \begin{bmatrix}
\Imat & &\bm* \\
\Ymat & \Imat&\bm*\\
& &\Imat_{3sn}\\
\end{bmatrix}.
$$
\end{corollary}
\begin{proof}
With (last) $3sn$ ancillary rows, we can be more efficient in parallelizing the construction presented in \Cref{lem:GenYCol}: We execute \Cref{lem:GenYCol} in parallel with $3sn$ ancillary rows. 
\Cref{fig:CNOTwithAncil1} helps to illustrate this construction.

\begin{figure*}[ht]
\centering
\scalebox{0.55}{\input{./pics/cnotancilla1.tex}}
\caption{The process to construct $\Ymat$ (for simplicity, we omit $2sn$ ancillary rows and corresponding columns in the construction of $\Ymat_{k},k\in[s]$).
}
\label{fig:CNOTwithAncil1}
\end{figure*}

We first divide $\Ymat$ as $\bmat{\Ymat_{1}&\cdots&\Ymat_{s}}$ where $\Ymat_i\in \Fbb_2^{n\times \log^2n}$.
By \Cref{lem:GenYCol}, we construct all $\Ymat_i$ simultaneously by $O(\log n)$ parallel row-elimination matrices with another $2sn$ ancillary rows.
Specifically, view the upper left $\Imat_{s\log^2 n}$ as $s$ blocks of $\Imat_{\log^2 n}$. 
We use the $i$-th $\Imat_{\log^2 n}$ and $3n$ rows to construct $\Ymat_i$, $n$ rows of which are for specifying $\Ymat_i$.
In other words, we put $\Ymat_i$ in the $((i-1)n +1)$-th to $in$-th rows as illustrated in \Cref{fig:CNOTwithAncil1}. 
We denote the whole process in this step by $\Rmat_{a}$.

We then add $\Ymat_{1},\ldots, \Ymat_{s}$ from the corresponding rows to the first $n$ ancillary rows to construct $\Ymat$. 
This can be done with $O(\log s)$ parallel row-elimination matrices, since we can add $s$ rows to any given row with $O(\log s)$ parallel row-elimination (See \Cref{prop:inarborescence}) and then restore all but the target rows by corresponding inverse operations.

Finally, we apply $\Rmat_{a}^{-1}$ to erase $\Ymat_{1},\ldots,\Ymat_{s}$ between the $(n+1)$-th and $sn$-th row.
\end{proof}

With the help of an illustration in \Cref{fig:CNOTwithAncil2}, we now complete the proof of \Cref{lem:M}, 

\begin{figure}[ht]
\centering
\scalebox{0.43}{\input{./pics/cnotancilla2.tex}}
\caption{The process to construct $\Mmat$ in the first $n$ ancillary rows. The last matrix omits all the ancillary rows.}
\label{fig:CNOTwithAncil2}
\end{figure}

\begin{proof}[Proof of \Cref{lem:M}] 
Let $t=s\log^2n$. 
We divide $\Mmat$ as $\bmat{\Ymat^{(1)}&\cdots&\Ymat^{(n/t)}}$ where each $\Ymat^{(i)}\in \Fbb_2^{n\times t}$.
Here we assume $n/t$ is an integer. If not, we can view the last sub-matrix $\Ymat^{(\ceilbra{ n/t })}$ as a $t$-column matrix by virtually padding zeros without ever changing their value. Especially, when we use \Cref{lem:GenYCol} and construct additive base matrix $\Pmat$, we keep the corresponding columns in $\Pmat$ to be $0$. 

Starting with $\Imat_{(3s+1)n}$, we sequentially construct $\Ymat^{(i)}$ with the $i$-th $\Imat_{t}$ and the rest $3sn$ rows.
This is shown in the first transformation of \Cref{fig:CNOTwithAncil2} by performing \Cref{cor:SParallel} repeatedly with $O\pbra{\frac{n}{t}\cdot\log n}=O\pbra{\frac n{s\log n}}$ parallel row-elimination matrices. 
After this step, we obtain 
$
\bmat{\Imat &&\\
\Mmat & \Imat &\bm*\\
&&\Imat_{3sn}}
$ 
as \Cref{lem:M} desires.
The time complexity is then $\widetilde{O}(sn\cdot \frac{n}{t})=\widetilde{O}(n^2)$ as stated in \Cref{lem:M}.
\end{proof}

\section{Lower Bound and Hardness Results}\label{sec:lowernpc}

\subsection{Lower Bound}

In this section, we present a depth lower bound for approximating CNOT circuits, establishing that our depth-and-ancillae trade-off remains asymptotically tight even if we relax the task from exact synthesis to approximation.
The lower bound also extends to synthesis with arbitrary two-qubit gates rather than only CNOT gates.

\begin{theorem*}[\Cref{thm:lowerbound} Restated]
Let $\epsilon<\sqrt2/2$ be a non-negative constant.
For $1-o(1)$ fraction of $n$-qubit CNOT circuits, any $\epsilon$-approximate $n$-qubit, $m$-ancilla quantum circuit has depth 
\begin{equation}\label{eq:LowerBoound}
\Omega\pbra{\max\cbra{\log n, \frac{n^2}{(n+m)\log (n+m)}}}.
\end{equation}
\end{theorem*}
\begin{proof}
We consider the two lower-bound terms in Equation \Cref{eq:LowerBoound} separately.

\paragraph*{The First Term.}
We first consider the following $n$-qubit CNOT circuit $\Ccal$:
$$
\Ccal\ket{x_1\cdots x_{n-1}x_n}=\ket{x_1\cdots x_{n-1}}\ket{x_1\oplus\cdots\oplus x_n}.
$$

To setup a proof by contradiction, suppose $\Ccal'$ is an $\epsilon$-approximate $n$-qubit, $m$-ancilla quantum circuit whose depth is $o(\log n)$. 
Due to the binary nature of the gates, this depth assumption implies that the $n$-th qubit in the output of $\Ccal'$ is only influenced by $o(n)$ inputs.
Assume without loss of generality that $\ket{x_1}$ is not among them. 
Let $\bm y=x_2\cdots x_n$.
By the $\epsilon$-approximation assumption, $\Ccal'\ket{0}\ket{\bm y}\ket{0}^{\otimes m}$ is $\epsilon$-close to $(\Ccal\ket{0}\ket{\bm y})\ket{0}^{\otimes m}$. 
Then, $\Ccal'\ket{1}\ket{\bm y}\ket{0}^{\otimes m}$ and $(\Ccal\ket{1}\ket{\bm y})\ket{0}^{\otimes m}$ will be near orthogonal. 
We have
$$
\vabs{\Ccal'\ket{1}\ket{\bm y}\ket{0}^{\otimes m}-(\Ccal\ket{1}\ket{\bm y})\ket{0}^{\otimes m}}_2>\sqrt2-\epsilon.
$$

For other CNOT circuits, this argument fails only when every output is affected by $o(n)$ inputs. 
Thus, the number of those circuits is upper bounded by
$$
\pbra{\sum_{k=o(n)}\binom{n}{k}}^n=\binom{n}{o(n)}^n.
$$
In contrast, the number of different CNOT circuits is
$$
\#\GL(n,2)=\prod_{i=0}^{n-1}\pbra{2^n-2^i}\geq2^{(n-1)^2/2}.
$$
Therefore, the fraction is upper bounded by
$$
2^{-(n-1)^2/2}\binom{n}{o(n)}^n=o(1).
$$

\paragraph*{The Second Term.} Any two-qubit gate $t$ can be uniquely described by $\Tmat_t\in\Cbb^{4\times 4}$ and expanded as a unitary matrix $\Umat_t\in\Cbb^{2^{n+m}\times2^{n+m}}$. 
Set
$$
d=o\pbra{\frac{n^2}{(n+m)\log (n+m)}},~
\delta=\frac{\sqrt2-2\epsilon}{4(n+m)d}
$$
with foresight.
We discretize all two-qubit gates into finitely many possibilities.

We define the \textit{$\delta$-discretization} of $\Umat_t$ as $\Umat^\delta_t$, where
$$
\Umat^\delta_t[u,v]=\delta\floorbra{\frac a\delta}+\delta\floorbra{\frac b\delta} i
$$
if $\Umat_t[u,v]=a+bi$. Then $\vabs{\Umat_t-\Umat_t^\delta}_2\leq2\delta$.
Moreover, although there are infinitely many two-qubit gates, there are only at most $\pbra{\frac2\delta}^{32}$ different $\delta$-discretization of two-qubit gates.

For any $n$-qubit, $m$-ancilla quantum circuit of depth $d$, there are $s<(n+m)d$ two-qubit gates. Also, the circuit can be viewed as a linear transform which sequentially multiplies $\Umat_{g_1},\ldots,\Umat_{g_s}$ to the input vector. Let $\Umat_{\bm g}=\Umat_{g_s}\cdots\Umat_{g_1}$ and $\Umat_{\bm g}^\delta=\Umat_{g_s}^\delta\cdots\Umat_{g_1}^\delta$.
Thus for any input vector $\bm v\in\Cbb^{2^{n+m}},\vabs{\bm v}_2=1$, we have
\begin{equation}
\vabs{\Umat_{\bm g}\bm v-\Umat_{\bm g}^\delta\bm v}_2< 2s\delta<\frac{\sqrt2}2-\epsilon. \label{eq:lowerbound}
\end{equation}

Assume two quantum circuits have gates $\cbra{h_i}_{i\in[s]},\cbra{t_i}_{i\in[s]}$ respectively, and $\Umat_{h_i}^\delta=\Umat_{t_i}^\delta$ holds for all $i\in[s]$.
Write $\Ccal_{\bm h}$ and $\Ccal_{\bm t}$ as the circuits corresponding to $\Umat_{\bm h}$ and $\Umat_{\bm t}$.
Suppose $\Ccal_{\bm h}$ and $\Ccal_{\bm t}$ are $\epsilon$-approximate to different $n$-qubit CNOT circuits $\Ccal_1$ and $\Ccal_2$. Then there exists $\bm x\in\bin^n$ that two circuits act differently on $\ket{\bm x}$, i.e.,
\begin{align*}
&\vabs{\Ccal_1\ket{\bm x}-\Ccal_2\ket{\bm x}}_2=\sqrt{2},\\
&\vabs{\Ccal_{\bm g}\ket{\bm x}\ket{0}^{\otimes m}-(\Ccal_1\ket{\bm x})\ket{0}^{\otimes m}}_2<\epsilon,\\
&\vabs{\Ccal_{\bm h}\ket{\bm x}\ket{0}^{\otimes m}-(\Ccal_2\ket{\bm x})\ket{0}^{\otimes m}}_2<\epsilon,
\end{align*}
which contradicts Inequality \Cref{eq:lowerbound} if setting $\bm v=\ket{\bm x}\ket{0}^{\otimes m}$.
We thus conclude that $\Umat_{\bm g},\Umat_{\bm h}$ can not approximate different CNOT circuits.

Consequently, the number of essentially different $n$-qubit $m$-ancilla quantum circuits of depth $d$ is upper bounded by
$$
\pbra{(n+m)\pbra{\frac2\delta}^{32}}^{d(n+m)}.
$$

Plugging in the parameters, the fraction is at most
\begin{equation*}
2^{-(n-1)^2/2}\pbra{\frac{2^{32}(n+m)}{\delta^{32}}}^{d(n+m)}=o(1).
\tag*{\qedhere}
\end{equation*}
\end{proof}

\subsection{Hardness Results}\label{sec:hardnessresult}

In this section, we study the computational hardness of CNOT-circuit optimization.
In our analyses, we will consider the following well-studied optimization problem:

\begin{problem}[$r$-Bounded Set Cover]
For any integer $r$, the problem, denoted by $\rBSC$, is defined as follows:
\begin{itemize}
\item Input: Positive integers $r,p,q$ and sets $W_1,\ldots,W_q\subseteq[p]$, such that (1) for all $i\in[q]$, $0<|W_i|\leq r$ and (2) $\bigcup_{i=1}^qW_i=[p]$.
\item Output: The minimum integer $k\geq 1$, for which there exists $V\subseteq[q]$ with $|V|=k$ such that $\bigcup_{i\in V}W_i=[p]$.
\end{itemize}
\end{problem}

The approximation complexity for this problem is well established.

\begin{theorem}[Inapproximability \cite{Feige98,Trevisan01}]\label{thm:setcovernph}
It is $\NP$-hard to approximate the solution of the $r$-Bounded Set Cover problem to within a factor of $\ln r-O(\ln\ln r)$.
\end{theorem}

Although there is an equivalence between CNOT circuit and $\GL(n,2)$, we introduce some notations to directly represent CNOT circuits to capture their topological constraints.

Let $\Gscr_n=\cbra{(i,j)\in[n]\times[n]\middle|i\neq j}$ be the set of all CNOT gates over $n$ qubits, where $(i,j)\in\Gscr_n$ specifies that
the $i$-th qubit controls the $j$-th.
For any $S\subseteq\Gscr_n$ (which we call \emph{topological constraints} of CNOT gates), let $\Lscr(S)$ be the set of all possible CNOT layers involving $S$: each layer $\ell\in\Lscr(S)$ is a subset of $S$ of non-interacting gates, i.e., any distinct $(i,j),(i',j')\in\ell$ satisfy $\{i,j\}\cap\{i',j'\}=\emptyset$.

For any sequence $g_1,\ldots,g_s\in S\subseteq\Gscr_n$, $\Ccal(g_1,\ldots,g_s)$ is the CNOT circuit wiring $g_i$ as the $i$-th gate. Similarly, for $\ell_1,\ldots,\ell_d\in\Lscr(S)\subseteq\Lscr(\Gscr_n)$, $\Ccal(\ell_1,\ldots,\ell_d)$ is the CNOT circuit whose $i$-th layer is $\ell_i$.

\begin{problem}[{\small Global Constrained Minimization}]
The problem, denoted by $\GCM$, is defined for depth and size respectively as follows:
\begin{itemize}
\item Input: Positive integers $n,m,s$, topological constraints $S\subseteq\Gscr_{n+m}$, and gates $g_1,\ldots,g_s\in S$ such that $\Ccal(g_1,\ldots,g_s)$ is an $n$-qubit $m$-ancilla CNOT circuit.
\item Output (size version): The minimum integer $u\geq0$ such that there exist gates $h_1,\ldots,h_u\in S$ and $\Ccal(g_1,\ldots,g_s)\cong\Ccal(h_1,\ldots,h_u)$.
\item Output (depth version): The minimum integer $v\geq0$ such that there exist layers $\ell_1,\ldots,\ell_v\subseteq\Lscr(S)$ and $\Ccal(g_1,\ldots,g_s)\cong\Ccal(\ell_1,\ldots,\ell_v)$.
\end{itemize}
\end{problem}

\begin{theorem}\label{thm:globalnph}
For both depth and size version of $\GCM$, it is $\NP$-hard to approximate the solution within any constant factor.
\end{theorem}

\begin{proof}
Given input $r,p,q,W_1,\ldots,W_q$ of $\rBSC$, we construct the input for $\GCM$ as follows:
\begin{itemize}
\item Set $n=p+c,m=q$ and $s=(c+2)p$, where $c$ is a positive integer to be determined later.
\item Set $S=\cbra{(i,j+n)\middle|i\in W_j,j\in[q]}\cup\{(j+n,p+t)|j\in[q],t\in[c]\}$.
\item For $i\in[p]$, let $t=(c+2)(i-1)$ and choose an arbitrary $j\in[q]$ such that $i\in W_j$. Then, set $g_{t+1}=(i,j+n)$, and $g_{t+t'+1}=(j+n,p+t')$ for each $t'\in[c]$, and $g_{t+c+2}=(i,j+n)$. 
\end{itemize}
Intuitively, if we denote the input as $\ket{x_1\cdots x_p}\ket{y_1\cdots y_c}\ket{0}^{\otimes m}$, then the $j$-th ancilla (which is the $(j+n)$-th qubit) represents $W_j$; the topological constraints $S$ only allows adding the $i$-th qubit to the $j$-th ancilla if $i\in W_j$ and allows adding ancillae to $y$'s. Thus the batch of $c+2$ gates during the construction with $i\in[p]$ essentially adds $x_i$ to all $y$'s, by first passing it to some $j$-th ancilla (i.e., $g_{t+1}$) and then adding it to $y$'s (i.e., $g_{t+t'+1},t'\in[c]$) and finally restoring it (i.e., $g_{t+c+2}$).
Therefore the output of the circuit is 
$$
\ket{x_1\cdots x_p}\ket{y_1\oplus x_1\oplus\cdots\oplus x_p}\cdots\ket{y_c\oplus x_1\oplus\cdots\oplus x_p}\ket{0}^{\otimes m}.
$$

Let $k$ be the solution to the $\rBSC$ instance. Let $u$ and $v$ be the solution to the size and depth version of the $\GCM$ instance respectively. We have the following claim which will be proved next.
\begin{claim}\label{clm:globalclm}
$k\leq v\leq kc+2r$ and $u=kc+2p$.
\end{claim}
Now we complete the proof: we set $c=1$ for the depth version and $c=p$ for the size version. Since $r$ is an arbitrary constant in $\rBSC$, \Cref{thm:globalnph} follows from \Cref{thm:setcovernph}.
\end{proof}

\begin{proof}[Proof of \Cref{clm:globalclm}]
We first focus on the upper bounds: $v\le kc+2r$ and $u\le kc+2p$. 
Assume without loss of generality $\bigcup_{i=1}^kW_i=[p]$ and define
$$
T_i=W_i\backslash\pbra{\bigcup_{j=1}^{i-1}W_j}=\cbra{t^i_1,\ldots,t^i_{|T_i|}}.
$$
Then we explicitly construct a circuit $\Ccal$ as follows:
\begin{itemize}
\item For $d\in[r]$, the $d$-th layer consists of gates $(t^i_d,i+n)$ for each $i\in[k]$ with $|T_i|\ge d$.

These $r$ layers ($p$ gates in total) prepare the $i$-th ancilla as the parity of $x$ on indices $T_i$.
\item For $t\in[c]$ and $j\in[k]$, the $(r+k(t-1)+j)$-th layer is simply one gate $(j+n,p+t)$.

These $kc$ layers ($kc$ gates in total) add the parity of ancillae (i.e., $x_1\oplus\cdots\oplus x_p$) to $y$'s.
\item For $d\in[r]$, the $(kc+r+d)$-th layer consists of gates $(t^i_d,i+n)$ for each $i\in[k]$ with $|T_i|\ge d$.

These $r$ layers ($p$ gates in total) restore the ancillae.
\end{itemize}
Thus $\Ccal$ has $kc+2r$ (possibly empty) layers and $kc+2p$ gates; and $\Ccal$ is essentially the same circuit as the input without violating the constraints.

We now establish the lower bound: $v\ge k$.
Assume $\Ccal$ is shallowest desired circuit.
Let $V$ be the gates in $\Ccal$ from ancillae to $y_1$, i.e., $j\in V$ iff gate $(j+n,p+1)\in\Ccal$.
Since these gates have the same target, we have $v\geq|V|$. 
Note that each $x_i$ can be added to $y_1$ only through these gates and thus $i\in\bigcup_{j\in V}W_j$. Hence $[p]=\bigcup_{j\in V}W_j$. By the definition of $k$, we have $|V|\geq k$ and $v\geq k$.

Finally we show $u\ge kc+2p$.
Assume $\Ccal$ is smallest desired circuit.
For any $j\in[c]$, let $V_j$ be the gates in $\Ccal$ from ancillae to $y_j$.
The analysis from last paragraph also applies here and shows $|V_j|\geq k$. 
Since each $x_i$ must be added to ancillae first in order to appear in $y_1,\ldots,y_c$, it must be added again to restore ancillae, which contributes at least $2p$ gates.
Thus $u\geq \sum_{j=1}^c|V_j|+2p\geq kc+2p$.
\end{proof}

The $\NP$-hardness analysis of $\GCM$ relies heavily on the existence of topological constraints. 
The current real-world implementation of quantum gates indeed has such constraints \cite{devitt2016performing,divincenzo2000physical}, though future quantum devices may be subject to fewer topological constraints.
Below, we introduce and analyze a ``local'' version of CNOT-circuit minimization where only a part of the circuit is allowed for optimization.

\begin{problem}[Local Size Minimization]
The problem, denoted by $\LSM$, is defined as follows:
\begin{itemize}
\item Input: Positive integers $n,m,s,L,R$ with $1\leq L\leq R\leq s$ and gates $g_1,\ldots,g_s\in\Gscr_{n+m}$ such that $\Ccal(g_1,\ldots,g_s)$ is an $n$-qubit $m$-ancilla CNOT circuit.
\item Output: The minimum integer $w\geq0$ such that there exist gates $h_1,\ldots,h_w\in\Gscr_{n+m}$ and
$$
\Ccal(g_1,\ldots,g_s)\cong\Ccal(g_1,\ldots,g_{L-1},h_1,\ldots,h_w,g_{R+1},\ldots,g_s).
$$
\end{itemize}
\end{problem}

\begin{theorem}\label{thm:localnph}
For $\LSM$, it is $\NP$-hard to approximate the solution within any constant factor.
\end{theorem}
\begin{proof}
Given input $r,p,q,W_1,\ldots,W_q$ of $\rBSC$, we construct the input for $\LSM$ as follows:
\begin{itemize}
\item Define $a_i=|W_i|\leq r$ for $i\in[q]$; and let $b=\sum_{i=1}^q2^{a_i}\le q2^r,c=\sum_{i=1}^q\sum_{j=0}^{a_i}j\cdot\binom{a_i}j\le qr2^r$.
\item Set $n=p+1,m=b$ and $L=c+1,R=c+p,s=2c+p$.
\item Name the $m$ ancillae by index $\text{ind}(i,U)\in[m]$ for each $i\in[q],U\subseteq W_i$.
\item The first $c$ gates: For $i\in[q],U\subseteq W_i$ and $j\in U$, construct gate $(j,n+\text{ind}(i,U))$.
\item The middle $p$ gates: For $i\in[p]$, construct gate $(i,n)$.
\item The final $c$ gates: For $i\in[q],U\subseteq W_i$ and $j\in U$, construct gate $(j,n+\text{ind}(i,U))$.
\end{itemize}
Intuitively if we denote the input as $\ket{x_1\cdots x_p}\ket{x_{p+1}}\ket{0}^{\otimes m}$, then the first $c$ gates prepare $\bigoplus_{j\in U}x_j$ on the $\text{ind}(i,U)$-th ancilla, and the middle $p$ gates update $\ket{x_{p+1}}$ to $\ket{x_{p+1}\oplus x_1\oplus\cdots\oplus x_p}$, and the final $c$ gates restore all the ancillae. 

Though the output of the circuit is simply $\ket{x_1\cdots x_p}\ket{x_{p+1}\oplus x_1\oplus\cdots\oplus x_p}\ket{0}^{\otimes m}$ where the first and final gates are redundant, the hardness of the problem comes from how to optimize the middle $p$ gates using the ancillae when we fix the first and final gates.

Let $k$ be the solution to the $\rBSC$ instance. Let $w$ be the solution to the $\LSM$ instance. Then the theorem follows immediately from \Cref{thm:setcovernph} and \Cref{clm:localclm} here.
\begin{claim}\label{clm:localclm}
$k=w$.\qedhere
\end{claim}
\end{proof}

\begin{proof}[Proof of \Cref{clm:localclm}]
We first show $w\leq k$.
Assume without loss of generality $\bigcup_{i=1}^kW_i=[p]$ and define
$$
T_i=W_i\backslash\pbra{\bigcup_{j=1}^{i-1}W_j}=\cbra{t_1^i,\ldots,t_{|T_i|}^i}.
$$
Since the $\text{ind}(i,T_i)$-th ancilla is already the parity of $x$ on indices $T_i$, we simply construct $h_i=(n+\text{ind}(i,T_i),n)$ for each $i\in[k]$. This also only updates $\ket{x_{p+1}}$ to $\ket{x_{p+1}\oplus x_1\oplus\cdots\oplus x_p}$ as the previous middle $p$ gates.

We now show $w\geq k$.
We view $\Ccal(h_1,\ldots,h_w)$ as a matrix $\Mmat$ in $\GL(n+m,2)$.
Let $\Ical$ be the set of ones in $\Mmat[n,*]$ (i.e., the $n$-th row of $\Mmat$). Then after $\Ccal(g_1,\ldots,g_{L-1},h_1,\ldots,h_w)$, $\ket{x_{p+1}}$ becomes
$$
\ket{\phi}=\ket{\bigoplus_{i\in\Ical,i\in[p+1]}x_i\oplus\bigoplus_{n+\text{ind}(i,U)\in\Ical}\bigoplus_{j\in U}x_j}.
$$
Since the final $c$ gates do not interact with the $n$-th qubit but $\ket{x_{p+1}}$ needs to be $\ket{x_{p+1}\oplus x_1\oplus\cdots\oplus x_p}$ in the end, we must have $\ket{\phi}=\ket{x_{p+1}\oplus x_1\oplus\cdots\oplus x_p}$.

Observe that $x_{p+1}$ is not involved in the first $c$ gates as well, thus $p+1\in\Ical$.
Meanwhile, since each CNOT gate is equivalent to a row elimination, any row in $\Mmat$ has at most $w+1$ ones; in particular, $|\Ical|\leq w+1$.
Now let $\Ecal=\Ical\backslash\cbra{p+1}=\cbra{e_1,\ldots,e_{w'}}$ where $w'=|\Ical|-1\le w$. For any $i\in[w]$,
\begin{itemize}
\item if $e_i\in[p]$, then choose any arbitrary $j\in[q]$ such that $e_i\in W_j$, and set $v_i=j$;
\item otherwise $e_i=n+\text{ind}(j,U)$ for some $j\in[q],U\subseteq W_j$, then set $v_i=j$.
\end{itemize}
Since $\ket{\phi}=\ket{x_{p+1}\oplus x_1\oplus\cdots\oplus x_p}$, we know
$$
[p+1]\subseteq\pbra{\Ical\cap[p+1]}\cup\bigcup_{n+\text{ind}(i,U)\in\Ical}U\subseteq[p+1].
$$
Therefore
$$
[p]
\subseteq\pbra{\Ical\cap[p]}\cup\bigcup_{n+\text{ind}(i,U)\in\Ical}U
\subseteq\pbra{\bigcup_{i\in\Ical\cap[p]}W_{v_i}}\cup\pbra{\bigcup_{n+\text{ind}(i,U)\in\Ical}W_i}
\subseteq[p],
$$
which means $W_{v_1},\ldots,W_{v_{w'}}$ form a cover of $[p]$ and thus $w\ge w'\ge k$ by the definition of $k$.
\end{proof}

\section{Conclusion and Open Questions}\label{sec:conclusion}

In this paper, we present efficient constructions for parallelizing CNOT circuits with or without ancillae, and provide information-theoretical lower bounds demonstrating that the space-depth trade-off of our constructions is asymptotically tight. 
In addition, we present strong evidences that CNOT-circuit optimization, both for size and depth, could be intractable:
\begin{itemize}
\item We prove that constructing equivalent CNOT-circuit with minimum depth (or size) under specified topologically-constrained qubit structure is $\NP$-hard to approximate within any constant factor. 
\item We show that the size-minimization of sub-circuits of CNOT circuits resists any efficient constant-ratio approximation as well.
\end{itemize}

While our algorithms are asymptotically optimal for parallel CNOT-circuit synthesis, many fundamental questions pertinent to theoretical/practical quantum logical synthesis remain open:

\vspace{0.1in}
\noindent{\bf Depth-Size Optimization}: Given the matrix representation of a \emph{specific} CNOT circuit, decide whether it can be
\begin{itemize}
\item parallelized to depth $d$, or
\item reduced to size $s$.
\end{itemize}
Are those problems $\NP$-complete without topological constraints of qubits, or with explicit constraints like 1D/2D grid?

\vspace{0.1in}
\noindent{\bf Characterization of Low-Depth Circuit}: Given specific structures of CNOT circuits, including being a CNOT tree as in \Cref{sec:noancilla}, can it be parallelized to $O(\log n)$ depth? More generally, can we characterize the class of all $O(\log n)$-depth or poly-logarithmic-depth CNOT circuits?

\vspace{0.1in}
\noindent{\bf CNOT Synthesis with Topological Constraints}: When there are topological constraints among qubits, is there any efficient algorithm to parallelize CNOT circuits? 

After the publication of this work at SODA'20, Wu \emph{et al.} \cite{wu2019optimization} designed an algorithm to optimize any given $n$-qubit CNOT circuit to $O(n^2/\log \delta)$ size when the topological constraint graph has minimum degree $\delta$.
For $d$-dimensional grid topological structure, they further optimized the CNOT circuit to $O(d\cdot n^{2/d})$ depth with fewer than $2n^2$ ancillae.

Can the algorithms in the recent work~\cite{nash2020quantum,wu2019optimization} be further improved?

\vspace{0.1in}
\noindent{\bf Practical Concerns}: 
To be more relevant to real-world applications, we give a simpler algorithm (see \Cref{prop:CNOT4nAdd6} in \Cref{app:simple}) for parallelizing any $n$-qubit CNOT circuit to depth $\leq3(n+1)$ without ancillae.
This construction---though not asymptotically optimal---could be useful for small-size circuits, because it does not carry a large overhead hidden in the big-O notation.

We remark that recently, De \emph{et al.}~\cite{de2021reducing} offered a construction of depth $\le4n/3 + 8\ceilbra{\log n}$, and Maslov \emph{et al.}~\cite{maslov2022depth} further improved it to $\floorbra{n + 1.9496\log^2 n + 3.5075\log n - 23.4269}$. 

Can we achieve asymptotic-optimality with a much simpler algorithm? Can our algorithms be realized as true parallel algorithms without the time-consuming pre-processing phase?

\section*{Acknowledgements}
We thank Hengjie Zhang and Yufan Zheng for valuable comments on the approximation of CNOT circuits.
We also thank Yuqing Kong and Shachar Lovett for inspiring ideas about the algorithms.

Shang-Hua Teng's research is supported by in part by the Simons Investigator Award for fundamental \& curiosity-driven research and NSF grant CCF1815254. Xiaoming Sun, Bujiao Wu and Jialin Zhang's research is supported in part by the Strategic Priority Research Program of Chinese Academy of Sciences Grant No. XDB28000000, and the National Natural Science Foundation of China Grants No. 61832003, 61872334. 

\bibliographystyle{alphaurl}
\bibliography{ref}

\appendix

\section{A simpler parallelizing algorithm}\label{app:simple}

\begin{proposition}\label{prop:CNOT4nAdd6}
There is an $O(n^3)$-time algorithm such that
given any $\Mmat \in \GL(n, 2)$,
it outputs parallel row-elimination matrices $\Rmat_1, \ldots, \Rmat_d$ where $d\leq3(n+1)$ such that $\Rmat_d\cdots\Rmat_2\Rmat_{1}\Mmat = \Imat$.
\end{proposition}

\begin{proof}
Similar with the proof of \Cref{lem:noancillaMatrix}, we first decompose $\Mmat = \Pmat\Lmat\Umat$,
where $\Lmat, \Umat$ are lower and upper triangular, and $\Pmat$ is a permutation matrix.
Since $\Pmat$ can be decomposed into $6$ parallel row-elimination matrices \cite{moore2001parallel},
it suffices to prove $\Lmat$ can be decomposed into at most $1.5(n-1)$ parallel row-elimination matrices, which is demonstrated in \Cref{alg:cnot4n}.

The intuition for \Cref{alg:cnot4n} is to use the main diagonal to eliminate each sub-diagonal with at most $2$ parallel row-eliminations. Take the third diagonal as an example: if $\Lmat[3,1]=\Lmat[4,2]=\Lmat[5,3]=\Lmat[6,4]=\Lmat[7,5]=1$, then we will add the CNOT gates $(1,3),(2,4),(3,5),(4,6),(5,7)$ to remove all ones on the third diagonal. This can be arranged as $2$ layers: $(1,3),(2,4),(5,7)$ and $(3,5),(4,6)$. Note that different arrangement may behave differently for the entries below the third diagonal, but they all serve the same purpose of cleaning the third diagonal while not interfering the entries above it. The general treatment can be found in \Cref{alg:cnot4n} below.

\begin{algorithm}[H]
\caption{Parallel row-elimination for $\Lmat$}\label{alg:cnot4n}
\For{$k\gets 1$ \KwTo $n-1$}{
	$\Tmat_0,\Tmat_1\gets\Imat$\;
	$s\gets k, b\gets0$\;
	\For{$i\gets1$ \KwTo $n-k$}{
		\lIf{$\Lmat[i + k,i] = 1$}{$\Tmat_b\gets\Rsfmat(i,i+k)\cdot\Tmat_1$}
		$s\gets s-1$\;
		\lIf{$s=0$}{$s\gets k,b\gets1-b$}
	}
	\lIf{$\Tmat_0\neq\Imat$}{\textbf{output} $\Tmat_0$}
	\lIf{$\Tmat_1\neq\Imat$}{\textbf{output} $\Tmat_1$}
	$\Lmat\gets\Tmat_1\Tmat_0\Lmat$\tcc*{Can be computed in time $O(n^2)$}
}
\end{algorithm}

Observe that in \Cref{alg:cnot4n}, when $n-k\le k$ we will never use $\Tmat_1$. Thus the total depth is bounded by
$$
\sum_{k<n/2}2+\sum_{n/2\le k\le n-1}1=
2\cdot\floorbra{\frac{n-1}2}+1\cdot\pbra{n-1-\floorbra{\frac{n-1}2}}\le\frac{3(n-1)}2
$$
as desired.
\end{proof}

\end{document}

%% file: pics/noancillaAlgo.tex
\begin{tikzpicture}
    \draw[black,dashed] (0,0) rectangle +(8,8);
    \draw[black,very thick] (0,8) -- (8,0);
    \fill[gray!80!white] (0,8) -- (8,0) -- (8,8) -- cycle;
    \node[scale=2.3] at (9.5,4) {$\xLongrightarrow{\!\!\text{Recursion}\!}$};
    \draw[black,very thick,dashed] (0,4) -- (8,4);
    \draw[black,very thick,dashed] (4,0) -- (4,8);
    \node[scale=3] at (2,6) {$\Mmat_1$};
    \node[scale=3] at (6,2) {$\Mmat_2$};
    \node[scale=3] at (6,6) {$\Amat$};
    
    \begin{scope}[xshift=11cm]
    \draw[black,dashed] (0,0) rectangle +(8,8);
    \draw[black,very thick] (0,8) -- (6,2) (7,1) -- (8,0);
    \fill[gray!80] (4,4) rectangle +(4,4);
    \foreach \x in {4,5,7}
        \draw[black,dashed,very thick] (\x,7-\x) rectangle +(1,1);
    \foreach \x in {6.25,6.5,6.75}
        \fill[black!80] (\x,8-\x) circle (2pt);
    \foreach \y in {4,5,6,7}{
        \draw[black,dashed] (4,\y) -- (8,\y);
    }
    \node[scale=1.8] at (6,7.5) {$\Amat_1$};
    \node[scale=1.8] at (6,6.5) {$\Amat_2$};
    \node[scale=1.8] at (6,4.5) {$\Amat_{\frac12\log n}$};
    \foreach \y in {5,6,7}
        \fill[black!80] (\y,5.5) circle (2pt);
    \node[scale=2.3] at (9.5,4) {$\xLongrightarrow{\text{Traverse}}$};
    \draw[decorate,decoration={brace,amplitude=5pt},xshift=-0.8mm] (4,7) -- (4,8) node[midway,xshift=-0.8cm,scale=1.5] {$\frac n{\log n}$};
    \draw[decorate,decoration={brace,amplitude=5pt},xshift=-0.8mm] (4,3) -- (4,4) node[midway,xshift=-1.1cm,scale=1.5] {$\frac12\!\log n$};
    \draw[decorate,decoration={brace,amplitude=5pt},yshift=-0.8mm] (5,3) -- (4,3) node[midway,yshift=-0.5cm,xshift=-0.4cm,scale=1.5] {$\frac12\!\log n$};
    \end{scope}

    \begin{scope}[xshift=22cm]
    \draw[black,dashed] (0,0) rectangle +(8,8);
    \draw[black,very thick] (0,8) -- (4,4);
    \fill[gray!80] (4,4) rectangle +(4,4);
    \foreach \x in {4,5,7}{
        \fill[gray!30] (\x,7-\x) rectangle +(1,1);
        \foreach \y in {1,2,3}
            \draw[dashed] (\x,7+0.25*\y-\x) -- +(1,0);
        \draw[-{Straight Barb[length=1mm]},xshift=-3mm] (\x+0.5,7.125-\x) to (\x+0.5,4.5);
        \draw[-{Straight Barb[length=1mm]},xshift=-1mm] (\x+0.5,7.125-\x+0.25) to (\x+0.5,5.5);
        \draw[-{Straight Barb[length=1mm]},xshift=1mm] (\x+0.5,7.125-\x+0.5) to (\x+0.5,6.5);
        \draw[-{Straight Barb[length=1mm]},xshift=3mm] (\x+0.5,7.125-\x+0.75) to (\x+0.5,7.5);
    }
    \foreach \x in {6.25,6.5,6.75}
        \fill[black!80] (\x,8-\x) circle (2pt);
    \foreach \y in {4,5,6,7}{
        \draw[black,dashed] (4,\y) -- (8,\y);
    }
    \foreach \y in {5,6,7}
        \fill[black!80] (\y,5.5) circle (2pt);
    \end{scope}

    \begin{scope}[yshift=-9cm]
    \draw[black,dashed] (0,0) rectangle +(8,8);
    \draw[black,very thick] (0,8) -- (6,2) (7,1) -- (8,0);
    \fill[gray!30] (4,4) rectangle +(4,4);
    \foreach \x in {4,5,7}
        \draw[dashed,very thick] (\x,7-\x) rectangle +(1,1);
    \foreach \x in {6.25,6.5,6.75}
        \fill[black!80] (\x,8-\x) circle (2pt);
    \foreach \y in {4,5,6,7}{
        \draw[black,dashed] (4,\y) -- (8,\y);
    }
    \node[scale=1.8] at (6,7.5) {$\Bmat_1$};
    \node[scale=1.8] at (6,6.5) {$\Bmat_2$};
    \node[scale=1.8] at (6,4.5) {$\Bmat_{\frac12\log n}$};
    \foreach \y in {5,6,7}
        \fill[black!80] (\y,5.5) circle (2pt);
    \node[scale=2.3] at (9.5,4) {$\xLongrightarrow{\text{Traverse}}$};
    \node[scale=2.3,yshift=-2mm] at (-1.5,4) {$\xLongrightarrow[\text{Layover}]{\text{Reach}}$};
    \end{scope}

    \begin{scope}[yshift=-9cm,xshift=11cm]
    \draw[black,dashed] (0,0) rectangle +(8,8);
    \draw[black,very thick] (0,8) -- (4,4);
    \fill[gray!30] (4,4) rectangle +(4,4);
    \foreach \x in {4,5,7}{
        \fill[gray!30] (\x,7-\x) rectangle +(1,1);
        \foreach \y in {1,2,3}
            \draw[dashed] (\x,7+0.25*\y-\x) -- +(1,0);
        \draw[-{Straight Barb[length=1mm]},xshift=-3mm] (\x+0.5,7.125-\x) to (\x+0.5,4.5);
        \draw[-{Straight Barb[length=1mm]},xshift=-1mm] (\x+0.5,7.125-\x+0.25) to (\x+0.5,5.5);
        \draw[-{Straight Barb[length=1mm]},xshift=1mm] (\x+0.5,7.125-\x+0.5) to (\x+0.5,6.5);
        \draw[-{Straight Barb[length=1mm]},xshift=3mm] (\x+0.5,7.125-\x+0.75) to (\x+0.5,7.5);
    }
    \foreach \x in {6.25,6.5,6.75}
        \fill[black!80] (\x,8-\x) circle (2pt);
    \foreach \y in {4,5,6,7}{
        \draw[black,dashed] (4,\y) -- (8,\y);
    }
    \foreach \y in {5,6,7}
        \fill[black!80] (\y,5.5) circle (2pt);
    \node[scale=2.3] at (9.5,4) {$\xLongrightarrow{\text{Done}}$};
    \end{scope}

    \begin{scope}[xshift=22cm,yshift=-9cm]
    \draw[black,dashed] (0,0) rectangle +(8,8);
    \draw[black,very thick] (0,8) -- (8,0);
    \end{scope}
\end{tikzpicture}

%% file: pics/generaln.tex
\begin{tikzpicture}
    \fill[gray!80] (4.5,3.5) rectangle +(3.5,4.5);
    \fill[gray!80] (4.5,3.5) rectangle +(3.5,0.5);
    \fill[gray!80] (7.5,3.5) rectangle +(0.5,4.5);
    \fill[pattern=north west lines, pattern color=gray!100] (4.5,4) rectangle +(3,4);
    \draw[black,dashed] (0,0) rectangle +(8,8);
    \draw[black,very thick] (0,8) -- (5.6,2.4) (6.4,1.6) -- (8,0);
    \foreach \x in {4.5,6.5}
        \draw[dashed,very thick] (\x,7-\x) rectangle +(1,1);
    \draw[densely dotted,very thick] (7.5,0) rectangle +(0.5,0.5);
    \foreach \x in {5.8,6,6.2}
        \fill[black!80] (\x,8-\x) circle (2pt);
    \foreach \y in {4,5,6,7}
        \draw[black,dashed] (4.5,\y) -- (7.5,\y);
    \draw[black,dashed] (7.5,4) -- (7.5,8);
    \foreach \y in {5.5,6,6.5}
        \fill[black!80] (\y,5.5) circle (2pt);
    \draw[decorate,decoration={brace,amplitude=5pt},xshift=-0.8mm] (4.5,2.5) -- (4.5,3.5) node[midway,xshift=-0.5cm,scale=1.5] {$g$};
    \draw[decorate,decoration={brace,amplitude=5pt},yshift=-0.8mm] (5.5,2.5) -- (4.5,2.5) node[midway,yshift=-0.5cm,scale=1.5] {$g$};
    \draw[decorate,decoration={brace,amplitude=5pt},yshift=0.8mm] (0,8) -- (4.5,8) node[midway,yshift=0.5cm,scale=1.5] {$\lceil n/2\rceil$};
    \draw[decorate,decoration={brace,amplitude=5pt},yshift=0.8mm] (4.5,8) -- (8,8) node[midway,yshift=0.5cm,scale=1.5] {$\lfloor n/2\rfloor$};
    \draw[decorate,decoration={brace,amplitude=4pt},xshift=0.8mm] (8,8) -- (8,7) node[midway,xshift=0.5cm,scale=1.5] {$h$};
    \draw[decorate,decoration={brace,amplitude=4pt},xshift=0.8mm] (8,4) -- (8,3.5) node[midway,xshift=0.5cm,scale=1.5] {$h'$};
    \draw[decorate,decoration={brace,amplitude=4pt},yshift=-0.8mm] (8,3.5) -- (7.5,3.5) node[midway,yshift=-0.4cm,scale=1.5] {$g'$};
    \draw[decorate,decoration={brace,amplitude=4pt},xshift=0.8mm] (8,0.5) -- (8,0) node[midway,xshift=0.5cm,scale=1.5] {$g'$};
    \draw[decorate,decoration={brace,amplitude=4pt},yshift=-0.8mm] (8,0) -- (7.5,0) node[midway,yshift=-0.4cm,scale=1.5] {$g'$};
    
    \begin{scope}[xshift=11cm]
    \fill[gray!80] (4.5,3.5) rectangle +(3.5,0.5);
    \fill[gray!80] (7.5,3.5) rectangle +(0.5,4.5);
    \fill[pattern=north west lines, pattern color=gray!100] (4.5,3.5) rectangle +(3,4);
    \draw[black,dashed] (0,0) rectangle +(8,8);
    \draw[black,very thick] (0,8) -- (5.6,2.4) (6.4,1.6) -- (8,0);
    \foreach \x in {4.5,6.5}
        \draw[dashed,very thick] (\x,7-\x) rectangle +(1,1);
    \draw[densely dotted,very thick] (7.5,0) rectangle +(0.5,0.5);
    \foreach \x in {5.8,6,6.2}
        \fill[black!80] (\x,8-\x) circle (2pt);
    \foreach \y in {4.5,5.5,6.5,7.5}
        \draw[black,dashed] (4.5,\y) -- (7.5,\y);
    \draw[black,dashed] (7.5,3.5) -- (7.5,7.5);
    \foreach \y in {5.5,6,6.5}
        \fill[black!80] (\y,5) circle (2pt);
    \draw[black,dashed] (4.5,3.5) -- (4.5,8) (4.5,3.5) -- (8,3.5);
    \draw[decorate,decoration={brace,amplitude=4pt},xshift=0.8mm] (8,4.5) -- (8,3.5) node[midway,xshift=0.5cm,scale=1.5] {$h$};
    \draw[decorate,decoration={brace,amplitude=5pt},xshift=-0.8mm] (4.5,2.5) -- (4.5,3.5) node[midway,xshift=-0.5cm,scale=1.5] {$g$};
    \draw[decorate,decoration={brace,amplitude=5pt},yshift=-0.8mm] (5.5,2.5) -- (4.5,2.5) node[midway,yshift=-0.5cm,scale=1.5] {$g$};
    \draw[decorate,decoration={brace,amplitude=5pt},yshift=0.8mm] (0,8) -- (4.5,8) node[midway,yshift=0.5cm,scale=1.5] {$\lceil n/2\rceil$};
    \draw[decorate,decoration={brace,amplitude=5pt},yshift=0.8mm] (4.5,8) -- (8,8) node[midway,yshift=0.5cm,scale=1.5] {$\lfloor n/2\rfloor$};
    \draw[decorate,decoration={brace,amplitude=4pt},xshift=0.8mm] (8,8) -- (8,7.5) node[midway,xshift=0.5cm,scale=1.5] {$h'$};
    \draw[decorate,decoration={brace,amplitude=4pt},yshift=-0.8mm] (8,3.5) -- (7.5,3.5) node[midway,yshift=-0.4cm,scale=1.5] {$g'$};
    \draw[decorate,decoration={brace,amplitude=4pt},xshift=0.8mm] (8,0.5) -- (8,0) node[midway,xshift=0.5cm,scale=1.5] {$g'$};
    \draw[decorate,decoration={brace,amplitude=4pt},yshift=-0.8mm] (8,0) -- (7.5,0) node[midway,yshift=-0.4cm,scale=1.5] {$g'$};
    \node[scale=2.3,yshift=-2mm] at (-1,4) {$\Longrightarrow$};
    \end{scope}
    
    \begin{scope}[xshift=22cm]
    \fill[gray!80] (7.5,3.5) rectangle +(0.5,4.5);
    \fill[pattern=north west lines, pattern color=gray!100] (5,4) rectangle +(3,4);
    \draw[black,dashed] (0,0) rectangle +(8,8);
    \draw[black,very thick] (0,8) -- (6.1,1.9) (6.9,1.1) -- (8,0);
    \foreach \x in {5,7}
        \draw[dashed,very thick] (\x,7-\x) rectangle +(1,1);
    \draw[densely dotted,very thick] (4.5,3) rectangle +(0.5,0.5);
    \foreach \x in {5.8,6,6.2}
        \fill[black!80] (\x+0.5,8-0.5-\x) circle (2pt);
    \foreach \y in {4,5,6,7}
        \draw[black,dashed] (5,\y) -- (8,\y);
    \foreach \y in {6,6.5,7}
        \fill[black!80] (\y,5.5) circle (2pt);
    \draw[black,dashed] (4.5,3.5) -- (4.5,8) (4.5,3.5) -- (8,3.5) (5,4) -- (5,8);
        
    \draw[decorate,decoration={brace,amplitude=5pt},xshift=0.8mm] (8,1) -- (8,0) node[midway,xshift=0.5cm,scale=1.5] {$g$};
    \draw[decorate,decoration={brace,amplitude=5pt},yshift=-0.8mm] (8,0) -- (7,0) node[midway,yshift=-0.5cm,scale=1.5] {$g$};
    \draw[decorate,decoration={brace,amplitude=5pt},yshift=0.8mm] (0,8) -- (4.5,8) node[midway,yshift=0.5cm,scale=1.5] {$\lceil n/2\rceil$};
    \draw[decorate,decoration={brace,amplitude=5pt},yshift=0.8mm] (4.5,8) -- (8,8) node[midway,yshift=0.5cm,scale=1.5] {$\lfloor n/2\rfloor$};
    \draw[decorate,decoration={brace,amplitude=4pt},yshift=-0.8mm] (8,3.5) -- (7.5,3.5) node[midway,yshift=-0.4cm,scale=1.5] {$g'$};
    \draw[decorate,decoration={brace,amplitude=4pt},xshift=-0.8mm] (4.5,3) -- (4.5,3.5) node[midway,xshift=-0.5cm,scale=1.5] {$g'$};
    \draw[decorate,decoration={brace,amplitude=4pt},yshift=-0.8mm] (5,3) -- (4.5,3) node[midway,yshift=-0.4cm,xshift=-0.1cm,scale=1.5] {$g'$};
    \draw[decorate,decoration={brace,amplitude=4pt},xshift=0.8mm] (8,8) -- (8,7) node[midway,xshift=0.5cm,scale=1.5] {$h$};
    \draw[decorate,decoration={brace,amplitude=4pt},xshift=0.8mm] (8,4) -- (8,3.5) node[midway,xshift=0.5cm,scale=1.5] {$h'$};
    \node[scale=2.3,yshift=-2mm] at (-1,4) {$\Longrightarrow$};
    \end{scope}
    
    \begin{scope}[yshift=-10cm,xshift=5.5cm]
    \fill[gray!80] (7.5,3.5) rectangle +(0.5,0.5);
    \fill[pattern=north west lines, pattern color=gray!100] (5,3.5) rectangle +(3,4);
    \draw[black,dashed] (0,0) rectangle +(8,8);
    \draw[black,very thick] (0,8) -- (6.1,1.9) (6.9,1.1) -- (8,0);
    \foreach \x in {5,7}
        \draw[dashed,very thick] (\x,7-\x) rectangle +(1,1);
    \draw[densely dotted,very thick] (4.5,3) rectangle +(0.5,0.5);
    \foreach \x in {5.8,6,6.2}
        \fill[black!80] (\x+0.5,8-0.5-\x) circle (2pt);
    \foreach \y in {4.5,5.5,6.5,7.5}
        \draw[black,dashed] (5,\y) -- (8,\y);
    \foreach \y in {6,6.5,7}
        \fill[black!80] (\y,5) circle (2pt);
    \draw[black,dashed] (4.5,3.5) -- (4.5,8) (4.5,3.5) -- (8,3.5) (5,3.5) -- (5,7.5);
        
    \draw[decorate,decoration={brace,amplitude=5pt},xshift=0.8mm] (8,1) -- (8,0) node[midway,xshift=0.5cm,scale=1.5] {$g$};
    \draw[decorate,decoration={brace,amplitude=5pt},yshift=-0.8mm] (8,0) -- (7,0) node[midway,yshift=-0.5cm,scale=1.5] {$g$};
    \draw[decorate,decoration={brace,amplitude=5pt},yshift=0.8mm] (0,8) -- (4.5,8) node[midway,yshift=0.5cm,scale=1.5] {$\lceil n/2\rceil$};
    \draw[decorate,decoration={brace,amplitude=5pt},yshift=0.8mm] (4.5,8) -- (8,8) node[midway,yshift=0.5cm,scale=1.5] {$\lfloor n/2\rfloor$};
    \draw[decorate,decoration={brace,amplitude=4pt},yshift=-0.8mm] (8,3.5) -- (7.5,3.5) node[midway,yshift=-0.4cm,scale=1.5] {$g'$};
    \draw[decorate,decoration={brace,amplitude=4pt},xshift=-0.8mm] (4.5,3) -- (4.5,3.5) node[midway,xshift=-0.5cm,scale=1.5] {$g'$};
    \draw[decorate,decoration={brace,amplitude=4pt},yshift=-0.8mm] (5,3) -- (4.5,3) node[midway,yshift=-0.4cm,xshift=-0.1cm,scale=1.5] {$g'$};
    \draw[decorate,decoration={brace,amplitude=4pt},xshift=0.8mm] (8,7.5) -- (8,6.5) node[midway,xshift=0.5cm,scale=1.5] {$h$};
    \draw[decorate,decoration={brace,amplitude=4pt},xshift=0.8mm] (8,4) -- (8,3.5) node[midway,xshift=0.5cm,scale=1.5] {$h'$};
    \node[scale=2.3,yshift=-2mm] at (-1,4) {$\Longrightarrow$};
    \end{scope}
    
    \begin{scope}[xshift=16.5cm,yshift=-10cm]
    \draw[black,dashed] (0,0) rectangle +(8,8);
    \draw[black,very thick] (0,8) -- (8,0);
    \node[scale=2.3,yshift=-2mm] at (-1,4) {$\Longrightarrow$};
    \end{scope}
\end{tikzpicture}

%% file: pics/outtreeexample.tex
\begin{tikzpicture}[cc/.style={draw,circle,inner sep=1.5pt}, node distance=1cm]
    \node (q1) at (0,0) {$1$};
    \node (q2) [below left=of q1,xshift=5mm] {$2$};
    \node (q3) [below right=of q1,xshift=-5mm] {$3$};
    \node (q4) [below left=of q2,xshift=5mm] {$4$};
    \node (q5) [below right=of q2,xshift=-5mm] {$5$};
    \draw[->,>=stealth,very thick] (q1) to (q2);
    \draw[->,>=stealth,very thick] (q1) to (q3);
    \draw[->,>=stealth,very thick] (q2) to (q4);
    \draw[->,>=stealth,very thick] (q2) to (q5);

    \begin{scope}[xshift=2.6cm,yshift=-1.4cm]
    \node[scale=2] at (0,0) {$\iff$};
    \node at (4.3,0) {
        \Qcircuit @C=2em @R=1.2em {
            \lstick{\ket{x_1}} & \ctrl{1} & \ctrl{2} & \qw & \qw & \qw\\
            \lstick{\ket{x_2}} & \targ & \qw & \ctrl{2} & \ctrl{3} & \qw\\
            \lstick{\ket{x_3}} & \qw & \targ & \qw & \qw & \qw\\
            \lstick{\ket{x_4}} & \qw & \qw & \targ & \qw & \qw\\
            \lstick{\ket{x_5}} & \qw & \qw & \qw & \targ & \qw 
        }
    };
    \end{scope}
\end{tikzpicture}

%% file: pics/intreeexample.tex
\begin{tikzpicture}[cc/.style={draw,circle,inner sep=1.5pt}, node distance=1cm]
    \node (q1) at (0,0) {$1$};
    \node (q2) [below left=of q1,xshift=5mm] {$2$};
    \node (q3) [below right=of q1,xshift=-5mm] {$3$};
    \node (q4) [below left=of q2,xshift=5mm] {$4$};
    \node (q5) [below right=of q2,xshift=-5mm] {$5$};
    \draw[->,>=stealth,very thick] (q2) to (q1);
    \draw[->,>=stealth,very thick] (q3) to (q1);
    \draw[->,>=stealth,very thick] (q4) to (q2);
    \draw[->,>=stealth,very thick] (q5) to (q2);

    \begin{scope}[xshift=2.6cm,yshift=-1.4cm]
    \node[scale=2] at (0,0) {$\iff$};
    \node at (4.3,0) {
        \Qcircuit @C=2em @R=1.2em {
            \lstick{\ket{x_1}} & \qw & \qw & \targ & \targ & \qw\\
            \lstick{\ket{x_2}} & \targ & \targ & \qw & \ctrl{-1} & \qw\\
            \lstick{\ket{x_3}} & \qw & \qw & \ctrl{-2} & \qw & \qw\\
            \lstick{\ket{x_4}} & \qw & \ctrl{-2} & \qw & \qw & \qw\\
            \lstick{\ket{x_5}} & \ctrl{-3} & \qw & \qw & \qw & \qw 
        }
    };
    \end{scope}
\end{tikzpicture}

%% file: pics/path.tex
\begin{tikzpicture}[cc/.style={draw,circle,inner sep=1.5pt}, node distance=1cm]
    \node (q1) at (0,0) {$1$};
    \node (q2) [right=of q1] {$2$};
    \node (qc) [right=of q2] {$\cdots$};
    \node (qn) [right=of qc] {$n$};
    \draw[->,>=stealth,very thick] (q1) to (q2);
    \draw[->,>=stealth,very thick] (q2) to (qc);
    \draw[->,>=stealth,very thick] (qc) to (qn);
\end{tikzpicture}

%% file: pics/intree.tex
\begin{tikzpicture}[cc/.style={draw,circle,inner sep=1.5pt}, node distance=1cm]
    \node (q1) at (0,0) {$1$};
    \node (qc) [below =of q1] {$\cdots$};
    \node (q3) [left=of qc] {$3$};
    \node (q2) [left=of q3] {$2$};
    \node (qnm1) [right=of qc] {$n-1$};
    \node (qn) [right=of qnm1] {$n$};
    \draw[->,>=stealth,very thick] (q2) to (q1);
    \draw[->,>=stealth,very thick] (q3) to (q1);
    \draw[->,>=stealth,very thick] (qnm1) to (q1);
    \draw[->,>=stealth,very thick] (qn) to (q1);
\end{tikzpicture}

%% file: pics/outtree.tex
\begin{tikzpicture}[cc/.style={draw,circle,inner sep=1.5pt}, node distance=1cm]
    \node (q1) at (0,0) {$1$};
    \node (qc) [below =of q1] {$\cdots$};
    \node (q3) [left=of qc] {$3$};
    \node (q2) [left=of q3] {$2$};
    \node (qnm1) [right=of qc] {$n-1$};
    \node (qn) [right=of qnm1] {$n$};
    \draw[->,>=stealth,very thick] (q1) to (q2);
    \draw[->,>=stealth,very thick] (q1) to (q3);
    \draw[->,>=stealth,very thick] (q1) to (qnm1);
    \draw[->,>=stealth,very thick] (q1) to (qn);
\end{tikzpicture}

%% file: pics/cnotcompress.tex
\begin{tikzpicture}[cc/.style={draw,circle,inner sep=1.5pt}, node distance=1cm]
    \node (q1) at (0,0) {$1$};
    \node (q2) [right=of q1] {$2$};
    \node (qc) [right=of q2] {$\cdots$};
    \node (qn) [right=of qc] {$k$};
    \draw[->,>=stealth,very thick] (q1) to (q2);
    \draw[->,>=stealth,very thick] (q2) to (qc);
    \draw[->,>=stealth,very thick] (qc) to (qn);
    \node (t1) [left=of q1] {$\Tcal_1$};
    \node (t2) [below=of t1] {$\Tcal_2$};
    \node (v) [right=of qn] {$v$};
    \node (t3) [below=of qn] {$\Tcal_3$};
    \draw[->,>=stealth,very thick] (t1) to (q1);
    \draw[->,>=stealth,very thick] (t2) to (q1);
    \draw[->,>=stealth,very thick] (qn) to (v);
    \draw[->,>=stealth,very thick] (t3) to (v);
    \node (cap) [below=of qc,yshift=-1cm] {$\text{In-arborescence }\Tcal$};
    \begin{scope}[xshift=3cm,yshift=-3.5cm]
        \node[scale=2] at (0,0) {$\Downarrow$};
    \end{scope}
    \begin{scope}[yshift=-4.5cm]
        \node (q11) at (0,0) {$1$};
        \node (qn1) [right=of q11] {$k$};
        \draw[->,>=stealth,very thick] (q11) to (qn1);
        \node (t11) [left=of q11] {$\Tcal_1$};
        \node (t21) [below=of t11] {$\Tcal_2$};
        \node (v1) [right=of qn1] {$v$};
        \node (t31) [below=of qn1] {$\Tcal_3$};
        \draw[->,>=stealth,very thick] (t11) to (q11);
        \draw[->,>=stealth,very thick] (t21) to (q11);
        \draw[->,>=stealth,very thick] (qn1) to (v1);
        \draw[->,>=stealth,very thick] (t31) to (v1);
        \node (cap1) [below=of t31,yshift=5mm] {$\text{In-arborescence }\Tcal'$};
    \end{scope}
    \begin{scope}[xshift=3.7cm,yshift=-5.5cm]
        \node[scale=2] at (0,0) {$+$};
    \end{scope}
    \begin{scope}[xshift=4.5cm,yshift=-5cm]
        \node (q12) at (0,0) {$1$};
        \node (q22) [right=of q12] {$2$};
        \node (qc2) [right=of q22] {$\cdots$};
        \node (qn2) [right=of qc2] {$k-1$};
        \draw[->,>=stealth,very thick] (q12) to (q22);
        \draw[->,>=stealth,very thick] (q22) to (qc2);
        \draw[->,>=stealth,very thick] (qc2) to (qn2);
        \node (cap2) [below=of qc2,yshift=-5mm] {$\text{In-arborescence }\Tcal''$};
    \end{scope}
\end{tikzpicture}

%% file: pics/ancillaprogress.tex
\begin{tikzpicture}
    \pgfmathsetmacro{\X}{1}
    \pgfmathsetmacro{\XX}{4*\X}
    \pgfmathsetmacro{\TXX}{3*\X}
    \pgfmathsetmacro{\Y}{1}
    \pgfmathsetmacro{\Z}{3}

    \draw[black,dashed] (0,0) rectangle +(3*\X,\TXX+\XX+\Y+\Z);
    \draw[black,dashed,very thick] (0,\Y+\Z+\XX) -- (3*\X,\Y+\Z+\XX);
    \draw[black,dashed,very thick] (0,\Z+\XX) -- (3*\X,\Z+\XX);
    \draw[black,dashed,very thick] (0,\Z) -- (3*\X,\Z);
    \draw[black,very thick] (0,\TXX+\XX+\Y+\Z) -- (2*\X,\TXX+\XX+\Y+\Z-2*\X);
    \foreach \x in {1,2,3}
        \fill[black!80] (2*\X+\x/4*\X,\TXX+\XX+\Y+\Z-2*\X-\x/4*\X) circle (2pt);
    \draw[decorate,decoration={brace,amplitude=6pt,raise=2pt}] (\X,\TXX+\XX+\Y+\Z) -- (\X,\TXX+\XX+\Y+\Z-\X)
        node[midway,xshift=1cm,scale=1.5] {$\frac12\!\log n$};
    \draw[decorate,decoration={brace,amplitude=6pt,raise=2pt}] (3*\X,\TXX+\XX+\Y+\Z) -- (3*\X,\TXX+\Y+\Z+\X)
        node[midway,xshift=1cm,scale=1.5] {$\log^2\!n$};
    \draw[decorate,decoration={brace,amplitude=7pt,raise=2pt}] (0,\TXX+\XX+\Y+\Z) -- (3*\X,\TXX+\XX+\Y+\Z)
        node[midway,yshift=0.7cm,scale=1.5] {$\log^2n$};
    \draw[decorate,decoration={brace,amplitude=6pt,raise=2pt}] (0,\XX+\Z) -- (0,\XX+\Y+\Z)
        node[midway,xshift=-0.5cm,scale=1.5] {$n$};
    \draw[decorate,decoration={brace,amplitude=7pt,raise=2pt}] (0,\Z) -- (0,\XX+\Z)
        node[midway,xshift=-0.5cm,scale=1.5] {$n$};
    \draw[decorate,decoration={brace,amplitude=7pt,raise=2pt}] (0,0) -- (0,\Z)
        node[midway,xshift=-0.5cm,scale=1.5] {$n$};
    \foreach \x in {0,1}
        \draw[black,dashed] (\x*\X,\TXX+\XX+\Y+\Z-\x*\X) rectangle +(\X,-\X);
    \node[scale=1.8,yshift=2pt] at (3*\X+1,\XX+\Y/2+\Z/2) {$\xLongrightarrow[\Pmat_i]{\substack{\text{Const-}\\\text{ruct}}}$};

    \begin{scope}[xshift=5cm]
    \draw[black,dashed] (0,0) rectangle +(3*\X,\TXX+\XX+\Y+\Z);
    \draw[black,dashed,very thick] (0,\Y+\Z+\XX) -- (3*\X,\Y+\Z+\XX);
    \draw[black,dashed,very thick] (0,\Z+\XX) -- (3*\X,\Z+\XX);
    \draw[black,dashed,very thick] (0,\Z) -- (3*\X,\Z);
    \draw[black,very thick] (0,\TXX+\XX+\Y+\Z) -- (2*\X,\TXX+\XX+\Y+\Z-2*\X);
    \foreach \x in {1,2,3}{
        \fill[black!80] (2*\X+\x/4*\X,\TXX+\XX+\Y+\Z-2*\X-\x/4*\X) circle (2pt);
        \fill[black!80] (2*\X+\X-\x/4*\X,\Z+\x/4*\XX/3) circle (2pt);
    }
    \foreach \x in {0,1}{
        \pgfmathsetmacro\tx{\x*\X}
        \pgfmathsetmacro\ty{\Z+\XX-\x*\XX/3}
        \pgfmathsetmacro\Ty{\TXX+\XX+\Y+\Z-\x*\X}
        \pgfmathtruncatemacro\result{\x+1}
        \draw[black,dashed] (\tx,\Ty) rectangle +(\X,-\X);
        \fill[gray!30] (\tx,\ty) rectangle +(\X,-\XX/3);
        \node[scale=1.5] at (\tx+\X/2,\ty-\XX/6) {$\Pmat_\result$};
        \draw[-{Straight Barb[length=1mm]}] (\tx+\X/2,\Ty-0.8*\X) to (\tx+\X/2,\ty-0.3*\X);
    }
    \draw[decorate,decoration={brace,amplitude=7pt,raise=2pt}] (2*\X,\XX+\Z-\XX/3) -- +(0,-\XX/3)
        node[midway,xshift=0.6cm,scale=1.2] {$\sqrt n$};
    \node[scale=1.8] at (3*\X+1,\XX+\Y/2+\Z/2) {$\xLongrightarrow[\Pmat_i]{\text{Copy}}$};
    \end{scope}

    \begin{scope}[xshift=10cm]
    \draw[black,dashed] (0,0) rectangle +(3*\X,\TXX+\XX+\Y+\Z);
    \draw[black,dashed,very thick] (0,\Y+\Z+\XX) -- (3*\X,\Y+\Z+\XX);
    \draw[black,dashed,very thick] (0,\Z+\XX) -- (3*\X,\Z+\XX);
    \draw[black,dashed,very thick] (0,\Z) -- (3*\X,\Z);
    \draw[black,very thick] (0,\TXX+\XX+\Y+\Z) -- (2*\X,\TXX+\XX+\Y+\Z-2*\X);
    \foreach \x in {1,2,3}{
        \fill[black!80] (2*\X+\x/4*\X,\TXX+\XX+\Y+\Z-2*\X-\x/4*\X) circle (2pt);
        \fill[black!80] (2*\X+\X-\x/4*\X,\Z+\x/4*\XX/3) circle (2pt);
        \pgfmathsetmacro\t{\x*(\Z-1.7)/4}
        \fill[black!80] (2*\X+\X-\x/4*\X,\t) circle (2pt);
    }
    \foreach \x in {0,1}{
        \pgfmathsetmacro\tx{\x*\X}
        \pgfmathsetmacro\ty{\Z+\XX-\x*\XX/3}
        \pgfmathsetmacro\Ty{\TXX+\XX+\Y+\Z-\x*\X}
        \pgfmathtruncatemacro\result{\x+1}
        \draw[black,dashed] (\tx,\Ty) rectangle +(\X,-\X);
        \fill[gray!30] (\tx,\ty) rectangle +(\X,-\XX/3);
        \node[scale=1.5] at (\tx+\X/2,\ty-\XX/6) {$\Pmat_\result$};
    }
    \node[scale=1.8,yshift=2.5pt] at (3*\X+1,\XX+\Y/2+\Z/2) {$\xLongrightarrow[\Ymat]{\substack{\text{Const-}\\\text{ruct}}}$};
    \fill[gray!30] (0,\Z) rectangle +(\X,-0.5);
    \draw[-{Straight Barb[length=1mm]}] (\X/2,\Z+2*\XX/3+0.2) to (\X/2,\Z-0.25);
    \fill[gray!30] (\X,\Z-0.5) rectangle +(\X,-1.2);
    \draw[-{Straight Barb[length=1mm]}] (\X/2+\X,\Z+\XX/3+0.2) to (\X/2+\X,\Z-1);
    \end{scope}

    \begin{scope}[xshift=15cm]
    \draw[black,dashed] (0,0) rectangle +(3*\X,\TXX+\XX+\Y+\Z);
    \fill[gray!30] (0,\XX+\Z) rectangle +(3*\X,\Y);
    \draw[black,dashed,very thick] (0,\Y+\Z+\XX) -- (3*\X,\Y+\Z+\XX);
    \draw[black,dashed,very thick] (0,\Z+\XX) -- (3*\X,\Z+\XX);
    \draw[black,dashed,very thick] (0,\Z) -- (3*\X,\Z);
    \draw[black,very thick] (0,\TXX+\XX+\Y+\Z) -- (2*\X,\TXX+\XX+\Y+\Z-2*\X);
    \foreach \x in {1,2,3}{
        \fill[black!80] (2*\X+\x/4*\X,\TXX+\XX+\Y+\Z-2*\X-\x/4*\X) circle (2pt);
        \fill[black!80] (2*\X+\X-\x/4*\X,\Z+\x/4*\XX/3) circle (2pt);
        \pgfmathsetmacro\t{\x*(\Z-1.7)/4}
        \fill[black!80] (2*\X+\X-\x/4*\X,\t) circle (2pt);
    }
    \foreach \x in {0,1}{
        \pgfmathsetmacro\tx{\x*\X}
        \pgfmathsetmacro\ty{\Z+\XX-\x*\XX/3}
        \pgfmathsetmacro\Ty{\TXX+\XX+\Y+\Z-\x*\X}
        \pgfmathtruncatemacro\result{\x+1}
        \draw[black,dashed] (\tx,\Ty) rectangle +(\X,-\X);
        \fill[gray!30] (\tx,\ty) rectangle +(\X,-\XX/3);
        \node[scale=1.5] at (\tx+\X/2,\ty-\XX/6) {$\Pmat_\result$};
    }
    \fill[gray!30] (0,\Z) rectangle +(\X,-0.5);
    \draw[-{Straight Barb[length=1mm]},bend left=6] (\X/4,\Z-0.25) to (\X/4,\Z+\XX+\Y/2-0.1);
    \draw[-{Straight Barb[length=1mm]},bend right=6] (3*\X/4,\XX+\Z-0.25) to (3*\X/4,\Z+\XX+\Y/2);
    \fill[gray!30] (\X,\Z-0.5) rectangle +(\X,-1.2);
    \draw[-{Straight Barb[length=1mm]}] (\X+\X/2+0.2,2*\XX/3+\Z-0.25) to (\X+\X/2+0.2,\Z+\XX+\Y/2-0.2);
    \draw[-{Straight Barb[length=1mm]},bend right=6] (\X+\X-0.1,\Z-1) to (\X+\X-0.1,\Z+\XX+\Y/2);
    \node[scale=1.5] at (1.5*\X,\XX+\Z+\Y/2) {$\Ymat$};
    \end{scope}

\end{tikzpicture}

%% file: pics/cnotancilla1.tex
\begin{tikzpicture}
    \pgfmathsetmacro{\X}{2.2}
    \pgfmathsetmacro{\T}{1.5/4*\X}
    \pgfmathsetmacro{\Z}{\X-\T}
    \draw[black,dashed] (0,\Z) rectangle +(2.5*\X+\T,2.5*\X+\T);
    \draw[black,very thick] (0,3.5*\X) -- (3.5*\X-\Z,\Z);
    \draw[black,dashed,very thick] (0,2.5*\X) -- (3.5*\X-\Z,2.5*\X);
    \draw[black,dashed,very thick] (0,1.5*\X+\Z) -- (3.5*\X-\Z,1.5*\X+\Z);
    \draw[black,dashed,very thick] (\X,3.5*\X) -- (\X,\Z); 
    \draw[decorate,decoration={brace,amplitude=10pt}] (\X,3.5*\X) -- (\X,2.5*\X) node[midway,xshift=1.2cm,scale=1.5] {$s\!\log^2n$};
    \draw[decorate,decoration={brace,amplitude=10pt}] (\X,1.5*\X+\Z) -- (\X,\Z) node[midway,xshift=0.7cm,scale=1.5] {$sn$};
    \draw[decorate,decoration={brace,amplitude=6pt,mirror}] (\X,2.5*\X) -- (\X,1.5*\X+\Z) node[midway,xshift=-0.5cm,scale=1.5] {$n$};
    \node[scale=2.3] at (2.5*\X+\T+0.75,1.25*\X+\T/2+\Z) {$\xLongrightarrow{\Rmat_a}$};

    \begin{scope}[xshift=3.55*\X cm]
    \draw[black,dashed] (0,\X-\T) rectangle +(2.5*\X+\T,2.5*\X+\T);
    \draw[black,very thick] (\X,2.5*\X) -- (3.5*\X-\Z,\Z);
    \draw[black,dashed,very thick] (0,2.5*\X) -- (3.5*\X-\Z,2.5*\X);
    \draw[black,dashed,very thick] (0,1.5*\X+\Z) -- (3.5*\X-\Z,1.5*\X+\Z);
    \draw[black,dashed,very thick] (\X,3.5*\X) -- (\X,\Z); 
    \draw[decorate,decoration={brace,amplitude=5pt,raise=1pt}] (\X,2.75*\X) -- (\X,2.5*\X) node[midway,xshift=1cm,scale=1.5,yshift=1mm] {$\log^2n$};
    \foreach \x in {0,1,2,3}{
        \pgfmathsetmacro\y{3-\x}
        \ifthenelse{\x=2}
            {
                \foreach \xx in {1,2,3}
                    \fill[black!80] (\x*\X/4+\xx*\X/16,3.5*\X-\x*\X/4-\xx*\X/16) circle (1pt);
            }
            {
                \fill[gray!30] (\x/4*\X,\X+\y/4*1.5*\X-\T) rectangle +(\X/4,1.5/4*\X);
                \draw[black,dashed] (\x*\X/4,3.5*\X-\x*\X/4) rectangle +(\X/4,-\X/4);
                \draw[black,very thick] (\x*\X/4,3.5*\X-\x*\X/4) -- +(\X/4,-\X/4);
                \draw[-{Straight Barb[length=1mm]}]
                    (\x*\X/4+\X/8,3.5*\X-\x*\X/4-\X/4*3/4) to (\x/4*\X+\X/8,\X+\y/4*1.5*\X+1.5/4*\X*6.5/8-\T);
            }
        \pgfmathtruncatemacro\result{\x+1}
        \node[scale=1.3] at (\x/4*\X+\X/8,\X+\y/4*1.5*\X+1.5*\X/8-\T)
            {
                \ifthenelse{\result<3}
                    {$\Ymat_{\!\result}$}
                    {\ifthenelse{\result<4}{}{$\Ymat_{\!s}$}}
            };
    }
    \foreach \x in {1,2,3}{
        \pgfmathsetmacro\u{\X/2+\x*\X/16}
        \pgfmathsetmacro\v{\X+1.5*\X/4*(2-\x/4)}
        \fill[black!80] (\u,\v-\T) circle (2pt);
        }
    \draw[decorate,decoration={brace,amplitude=6pt,raise=2pt}] (\X,\X+\X*1.5/4-\T) -- (\X,\X-\T) node[midway,xshift=0.5cm,scale=1.5] {$n$};
    \node[scale=1.8] at (2.5*\X+\T+0.75,1.25*\X+\T/2+\Z) {$\xLongrightarrow{\!\!\textbf{Add}\!}$};
    \end{scope}

    \begin{scope}[yshift=-3.2*\X cm]
    \draw[black,dashed] (0,\X-\T) rectangle +(2.5*\X+\T,2.5*\X+\T);
    \draw[black,very thick] (0,3.5*\X) -- (3.5*\X-\Z,\Z);
    \draw[black,dashed,very thick] (0,1.5*\X+\Z) -- (3.5*\X-\Z,1.5*\X+\Z);
    \draw[black,dashed,very thick] (0,2.5*\X) -- (3.5*\X-\Z,2.5*\X);
    \draw[black,dashed,very thick] (\X,3.5*\X) -- (\X,\Z);
    \foreach \x in {0,1,2,3}{
        \pgfmathsetmacro\y{3-\x}
        \ifthenelse{\x=2}{}
            {
                \fill[gray!30] (\x/4*\X,\X+\y/4*1.5*\X-\T) rectangle +(\X/4,1.5/4*\X);
                \fill[gray!30] (\x/4*\X,1.5*\X+\Z) rectangle +(\X/4,1.5/4*\X);
                \draw[-{Straight Barb[length=1mm]}]
                    (\x/4*\X+\X/8,\X+\y/4*1.5*\X+1.5/4*\X-1.5*\X/4*9/8) to (\x/4*\X+\X/8,1.5*\X+\Z+1.5/4*\X*1.5/8);
            }
        \pgfmathtruncatemacro\result{\x+1}
        \node[scale=1.3] at (\x/4*\X+\X/8,\X+\y/4*1.5*\X+1.5*\X/8-\T)
            {
                \ifthenelse{\result<3}
                    {$\Ymat_{\!\result}$}
                    {\ifthenelse{\result<4}{}{$\Ymat_{\!s}$}}
            };
        \node[scale=1.3] at (\x/4*\X+\X/8,\Z+1.5*\X/8+1.5*\X)
            {
                \ifthenelse{\result<3}
                    {$\Ymat_{\!\result}$}
                    {\ifthenelse{\result<4}{}{$\Ymat_{\!s}$}}
            };
    }
    \foreach \x in {1,2,3}{
        \pgfmathsetmacro\u{\X/2+\x*\X/16}
        \pgfmathsetmacro\v{\X+1.5*\X/4*(2-\x/4)}
        \pgfmathsetmacro\vv{\X+1.5*\X/8}
        \fill[black!80] (\u,\v-\T) circle (2pt);
        \fill[black!80] (\u,\vv-\T+1.5*\X) circle (1.2pt);
    }
    \fill[gray!80] (\X+\T,\Z+1.5*\X) rectangle +(1.5*\X,\T);
    \node[scale=2.3,yshift=1.5pt] at (2.5*\X+\T+0.75,1.25*\X+\T/2+\Z) {$\xLongrightarrow{\!\!\Rmat^{-1}_a\!\!}$};
    \end{scope}

    \begin{scope}[xshift=3.55*\X cm, yshift=-3.2*\X cm]
    \draw[black,dashed] (0,\X-\T) rectangle +(2.5*\X+\T,2.5*\X+\T);
    \draw[black,very thick] (0,3.5*\X) -- (3.5*\X-\Z,\Z);
    \draw[black,dashed,very thick] (0,1.5*\X+\Z) -- (3.5*\X-\Z,1.5*\X+\Z);
    \draw[black,dashed,very thick] (0,2.5*\X) -- (3.5*\X-\Z,2.5*\X);
    \draw[black,dashed,very thick] (\X,3.5*\X) -- (\X,\Z);
    \fill[gray!30] (0,\Z+1.5*\X) rectangle +(\X,\T);
    \node[scale=1.4] at (\X/2,\Z+\T/2+1.5*\X) {$\Ymat$};
    \fill[gray!80] (\X+\T,\Z+1.5*\X) rectangle +(1.5*\X,\T);
    \end{scope}

\end{tikzpicture}

%% file: pics/cnotancilla2.tex
\begin{tikzpicture}
    \pgfmathsetmacro{\X}{2.2}
    \pgfmathsetmacro{\T}{1.5/4*\X}
    \pgfmathsetmacro{\Z}{\X-\T}
    \draw[black,dashed] (0,\Z) rectangle +(2.5*\X+\T,2.5*\X+\T);
    \draw[black,very thick] (0,3.5*\X) -- (3.5*\X-\Z,\Z);
    \draw[black,dashed,very thick] (0,2.5*\X) -- (3.5*\X-\Z,2.5*\X); 
    \draw[black,dashed,very thick] (0,2.5*\X-\T) -- (3.5*\X-\Z,2.5*\X-\T);
    \draw[black,dashed,very thick] (\X,3.5*\X) -- (\X,\Z); 
    \draw[decorate,decoration={brace,amplitude=10pt}] (\X,3.5*\X) -- (\X,2.5*\X) node[midway,xshift=0.7cm,scale=1.5] {$n$};
    \draw[decorate,decoration={brace,amplitude=10pt}] (\X,2.5*\X-\T) -- (\X,\X-\T) node[midway,xshift=0.9cm,scale=1.5] {$3sn$};
    \draw[decorate,decoration={brace,amplitude=6pt,mirror}] (\X,\X+1.5*\X) -- (\X,\Z+1.5*\X) node[midway,xshift=-0.5cm,scale=1.5] {$n$};
    \fill[gray!30] (0,\Z+1.5*\X) rectangle +(\X/3,\T);
    \node[scale=1.1] at (\X/6,\Z+\T/2+1.5*\X) {$\Ymat^{(1)}$};
    \node[scale=2.3] at (2.5*\X+\T+0.75,1.25*\X+\T/2+\Z) {$\xLongrightarrow{\cdots}$};
    \fill[gray!80] (\X+\T,2.5*\X-\T) rectangle +(1.5*\X,\T);

    \begin{scope}[xshift=3.55*\X cm]
    \draw[black,dashed] (0,\X-\T) rectangle +(2.5*\X+\T,2.5*\X+\T);
    \draw[black,very thick] (0,3.5*\X) -- (3.5*\X-\Z,\Z);
    \draw[black,dashed,very thick] (0,2.5*\X) -- (3.5*\X-\Z,2.5*\X);
    \draw[black,dashed,very thick] (0,2.5*\X-\T) -- (3.5*\X-\Z,2.5*\X-\T);
    \draw[black,dashed,very thick] (\X,3.5*\X) -- (\X,\Z);
    \fill[gray!30] (0,\Z+1.5*\X) rectangle +(\X,\T);
    \node[scale=1.3] at (\X/2,\Z+\T/2+1.5*\X) {$\Mmat$};
    \fill[gray!80] (\X+\T,\Z+1.5*\X) rectangle +(1.5*\X,\T);
    \node[scale=2.3,yshift=-1mm] at (2.5*\X+\T+0.5,1.25*\X+\T/2+\Z) {$\Rightarrow$};
    \end{scope}

    \begin{scope}[xshift=6.9*\X cm,yshift=1.8cm,scale=0.6]
    \draw[black,dashed] (0,\X-\T) rectangle +(2.5*\X+\T,2.5*\X+\T);
    \draw[black,very thick] (0,3.5*\X) -- (3.5*\X-\Z,\Z);
    \draw[black,dashed,very thick] (0,1.75*\X+0.5*\Z) -- (3.5*\X-\Z,1.75*\X+0.5*\Z)
        (1.25*\X+0.5*\T,\Z) -- (1.75*\X-0.5*\Z,3.5*\X);
    \fill[gray!30] (0,\Z) rectangle +(1.75*\X-0.5*\Z,1.75*\X-0.5*\Z);
    \node[scale=2] at (1.75*\X/2-0.25*\Z,1.75*\X/2-0.25*\Z+\Z) {$\Mmat$};
    \draw[decorate,decoration={brace,amplitude=10pt}] (1.25*\X+0.5*\T,3.5*\X) -- (1.25*\X+0.5*\T,1.75*\X+0.5*\Z)
        node[midway,xshift=0.6cm,scale=1.5] {$n$};
    \draw[decorate,decoration={brace,amplitude=10pt}] (1.25*\X+0.5*\T,1.75*\X+0.5*\Z) -- (1.25*\X+0.5*\T,\Z)
        node[midway,xshift=0.6cm,scale=1.5] {$n$};
    \end{scope}
\end{tikzpicture}

%% file: main.bbl
\newcommand{\etalchar}[1]{$^{#1}$}
\begin{thebibliography}{KMvdG20}

\bibitem[AG04]{aaronson2004improved}
Scott Aaronson and Daniel Gottesman.
\newblock Improved simulation of stabilizer circuits.
\newblock {\em Physical Review A}, 70(5):052328, 2004.

\bibitem[AHM07]{andren2007complexity}
Daniel Andr{\'e}n, Lars Hellstr{\"o}m, and Klas Markstr{\"o}m.
\newblock On the complexity of matrix reduction over finite fields.
\newblock {\em Advances in applied mathematics}, 39(4):428--452, 2007.

\bibitem[Alo03]{Alon03}
Noga Alon.
\newblock A simple algorithm for edge-coloring bipartite multigraphs.
\newblock {\em Inf. Process. Lett.}, 85(6):301--302, 2003.
\newblock \href {https://doi.org/10.1016/S0020-0190(02)00446-5}
  {\path{doi:10.1016/S0020-0190(02)00446-5}}.

\bibitem[AMM14]{amy2014polynomial}
Matthew Amy, Dmitri Maslov, and Michele Mosca.
\newblock Polynomial-time t-depth optimization of clifford+ t circuits via
  matroid partitioning.
\newblock {\em IEEE Transactions on Computer-Aided Design of Integrated
  Circuits and Systems}, 33(10):1476--1489, 2014.

\bibitem[AW21]{AlmanW21}
Josh Alman and Virginia~Vassilevska Williams.
\newblock A refined laser method and faster matrix multiplication.
\newblock In D{\'{a}}niel Marx, editor, {\em Proceedings of the 2021 {ACM-SIAM}
  Symposium on Discrete Algorithms, {SODA} 2021, Virtual Conference, January 10
  - 13, 2021}, pages 522--539. {SIAM}, 2021.
\newblock \href {https://doi.org/10.1137/1.9781611976465.32}
  {\path{doi:10.1137/1.9781611976465.32}}.

\bibitem[BBC{\etalchar{+}}95]{barenco1995elementary}
Adriano Barenco, Charles~H Bennett, Richard Cleve, David~P DiVincenzo, Norman
  Margolus, Peter Shor, Tycho Sleator, John~A Smolin, and Harald Weinfurter.
\newblock Elementary gates for quantum computation.
\newblock {\em Physical review A}, 52(5):3457, 1995.

\bibitem[BEH74]{LUdecom}
James Bunch and John E.~Hopcroft.
\newblock Triangular factorization and inversion by fast matrix multiplication.
\newblock {\em Mathematics of Computation - Math. Comput.}, 28:231--231, 01
  1974.
\newblock \href {https://doi.org/10.2307/2005828} {\path{doi:10.2307/2005828}}.

\bibitem[Chr14]{christofides2014asymptotic}
Demetres Christofides.
\newblock The asymptotic complexity of matrix reduction over finite fields.
\newblock {\em arXiv preprint arXiv:1406.5826}, 2014.

\bibitem[CMRT88]{cosnard1988parallel}
Michel Cosnard, Mounir Marrakchi, Yves Robert, and Denis Trystram.
\newblock Parallel gaussian elimination on an mimd computer.
\newblock {\em Parallel Computing}, 6(3):275--296, 1988.

\bibitem[COS01]{ColeOS01}
Richard Cole, Kirstin Ost, and Stefan Schirra.
\newblock Edge-coloring bipartite multigraphs in {O(E} log {D)} time.
\newblock {\em Combinatorica}, 21(1):5--12, 2001.
\newblock \href {https://doi.org/10.1007/s004930170002}
  {\path{doi:10.1007/s004930170002}}.

\bibitem[CSU19]{childs2019circuit}
Andrew~M Childs, Eddie Schoute, and Cem~M Unsal.
\newblock Circuit transformations for quantum architectures.
\newblock In {\em 14th Conference on the Theory of Quantum Computation,
  Communication and Cryptography}, 2019.

\bibitem[CW00]{cleve2000fast}
Richard Cleve and John Watrous.
\newblock Fast parallel circuits for the quantum fourier transform.
\newblock In {\em Proceedings 41st Annual Symposium on Foundations of Computer
  Science}, pages 526--536. IEEE, 2000.

\bibitem[dBBV{\etalchar{+}}21]{de2021reducing}
Timoth{\'e}e~Goubault de~Brugi{\`e}re, Marc Baboulin, Benoit Valiron, Simon
  Martiel, and Cyril Allouche.
\newblock Reducing the depth of linear reversible quantum circuits.
\newblock {\em IEEE Transactions on Quantum Engineering}, 2:1--22, 2021.

\bibitem[Dev16]{devitt2016performing}
Simon~J Devitt.
\newblock Performing quantum computing experiments in the cloud.
\newblock {\em Physical Review A}, 94(3):032329, 2016.

\bibitem[DiV95]{divincenzo1995two}
David~P DiVincenzo.
\newblock Two-bit gates are universal for quantum computation.
\newblock {\em Physical Review A}, 51(2):1015, 1995.

\bibitem[DiV00]{divincenzo2000physical}
David~P DiVincenzo.
\newblock The physical implementation of quantum computation.
\newblock {\em Fortschritte der Physik: Progress of Physics},
  48(9-11):771--783, 2000.

\bibitem[Fei98]{Feige98}
Uriel Feige.
\newblock A threshold of ln \emph{n} for approximating set cover.
\newblock {\em J. {ACM}}, 45(4):634--652, 1998.
\newblock \href {https://doi.org/10.1145/285055.285059}
  {\path{doi:10.1145/285055.285059}}.

\bibitem[FL10]{faugere2010parallel}
Jean-Charles Faug{\`e}re and Sylvain Lachartre.
\newblock Parallel gaussian elimination for gr{\"o}bner bases computations in
  finite fields.
\newblock In {\em Proceedings of the 4th International Workshop on Parallel and
  Symbolic Computation}, pages 89--97. ACM, 2010.

\bibitem[Gal14]{Gall14a}
Fran{\c{c}}ois~Le Gall.
\newblock Powers of tensors and fast matrix multiplication.
\newblock In {\em International Symposium on Symbolic and Algebraic
  Computation, {ISSAC} '14, Kobe, Japan, July 23-25, 2014}, pages 296--303,
  2014.
\newblock \href {https://doi.org/10.1145/2608628.2608664}
  {\path{doi:10.1145/2608628.2608664}}.

\bibitem[HND17]{herr2017optimization}
Daniel Herr, Franco Nori, and Simon~J Devitt.
\newblock Optimization of lattice surgery is np-hard.
\newblock {\em npj Quantum Information}, 3(1):35, 2017.

\bibitem[KC00]{kabanets2000circuit}
Valentine Kabanets and Jin-Yi Cai.
\newblock Circuit minimization problem.
\newblock In {\em Proceedings of the thirty-second annual ACM symposium on
  Theory of computing}, pages 73--79. ACM, 2000.

\bibitem[KKT91]{KaklamanisKT91}
Christos Kaklamanis, Danny Krizanc, and Thanasis Tsantilas.
\newblock Tight bounds for oblivious routing in the hypercube.
\newblock {\em Mathematical Systems Theory}, 24(4):223--232, 1991.
\newblock \href {https://doi.org/10.1007/BF02090400}
  {\path{doi:10.1007/BF02090400}}.

\bibitem[KMvdG20]{kissinger2020cnot}
A~Kissinger and A~Meijer-van~de Griend.
\newblock Cnot circuit extraction for topologically-constrained quantum
  memories.
\newblock {\em Quantum Information and Computation}, 20:581--596, 2020.

\bibitem[Kni95]{knill1995approximation}
Emanuel Knill.
\newblock Approximation by quantum circuits.
\newblock {\em arXiv preprint quant-ph/9508006}, 1995.

\bibitem[Lov17]{Lovett17}
Shachar Lovett.
\newblock Additive combinatorics and its applications in theoretical computer
  science.
\newblock {\em Theory of Computing, Graduate Surveys}, 8:1--55, 2017.
\newblock \href {https://doi.org/10.4086/toc.gs.2017.008}
  {\path{doi:10.4086/toc.gs.2017.008}}.

\bibitem[MN01]{moore2001parallel}
Cristopher Moore and Martin Nilsson.
\newblock Parallel quantum computation and quantum codes.
\newblock {\em SIAM Journal on Computing}, 31(3):799--815, 2001.

\bibitem[MR85]{miller1985parallel}
Gary~L. Miller and John~H. Reif.
\newblock Parallel tree contraction and its application.
\newblock In {\em 26th Annual Symposium on Foundations of Computer Science},
  pages 478--489, 1985.

\bibitem[MR89]{miller1989parallel}
Gary~L. Miller and John~H. Reif.
\newblock Parallel tree contraction part 1: Fundamentals.
\newblock {\em Advances in Computing Research}, 5:47--72, 1989.

\bibitem[MR91]{miller1991parallel}
Gary~L Miller and John~H Reif.
\newblock Parallel tree contraction part 2: further applications.
\newblock {\em SIAM Journal on Computing}, 20(6):1128--1147, 1991.

\bibitem[MZ22]{maslov2022depth}
Dmitri Maslov and Ben Zindorf.
\newblock Depth optimization of cz, cnot, and clifford circuits, 2022.
\newblock \href {http://arxiv.org/abs/2201.05215} {\path{arXiv:2201.05215}}.

\bibitem[NGM20]{nash2020quantum}
Beatrice Nash, Vlad Gheorghiu, and Michele Mosca.
\newblock Quantum circuit optimizations for nisq architectures.
\newblock {\em Quantum Science and Technology}, 5(2):025010, 2020.

\bibitem[PMH08]{patel2008optimal}
Ketan~N Patel, Igor~L Markov, and John~P Hayes.
\newblock Optimal synthesis of linear reversible circuits.
\newblock {\em Quantum Information \& Computation}, 8(3):282--294, 2008.

\bibitem[Pre18]{preskill2018quantum}
John Preskill.
\newblock Quantum computing in the nisq era and beyond.
\newblock {\em Quantum}, 2:79, 2018.

\bibitem[RPP{\etalchar{+}}06]{rupp2006parallel}
A~Rupp, J~Pelzl, Christof Paar, MC~Mertens, and A~Bogdanov.
\newblock A parallel hardware architecture for fast gaussian elimination over
  gf (2).
\newblock In {\em 2006 14th Annual IEEE Symposium on Field-Programmable Custom
  Computing Machines}, pages 237--248. IEEE, 2006.

\bibitem[RR97]{razborov1997natural}
Alexander~A Razborov and Steven Rudich.
\newblock Natural proofs.
\newblock {\em Journal of Computer and System Sciences}, 55(1):24--35, 1997.

\bibitem[RT89]{robert1989optimal}
Yves Robert and Denis Trystram.
\newblock Optimal scheduling algorithms for parallel gaussian elimination.
\newblock {\em Theoretical Computer Science}, 64(2):159--173, 1989.

\bibitem[SBM06]{ShendeBM06}
Vivek~V. Shende, Stephen~S. Bullock, and Igor~L. Markov.
\newblock Synthesis of quantum-logic circuits.
\newblock {\em {IEEE} Trans. on {CAD} of Integrated Circuits and Systems},
  25(6):1000--1010, 2006.
\newblock \href {https://doi.org/10.1109/TCAD.2005.855930}
  {\path{doi:10.1109/TCAD.2005.855930}}.

\bibitem[Sel13]{selinger2013quantum}
Peter Selinger.
\newblock Quantum circuits of t-depth one.
\newblock {\em Physical Review A}, 87(4):042302, 2013.

\bibitem[SMB04]{shende2004smaller}
Vivek~V Shende, Igor~L Markov, and Stephen~S Bullock.
\newblock Smaller two-qubit circuits for quantum communication and computation.
\newblock In {\em Proceedings Design, Automation and Test in Europe Conference
  and Exhibition}, volume~2, pages 980--985. IEEE, 2004.

\bibitem[Str69]{Strassen1969}
Volker Strassen.
\newblock Gaussian elimination is not optimal.
\newblock {\em Numerische Mathematik}, 13(4):354--356, Aug 1969.
\newblock \href {https://doi.org/10.1007/BF02165411}
  {\path{doi:10.1007/BF02165411}}.

\bibitem[Tre01]{Trevisan01}
Luca Trevisan.
\newblock Non-approximability results for optimization problems on bounded
  degree instances.
\newblock In {\em Proceedings on 33rd Annual {ACM} Symposium on Theory of
  Computing, July 6-8, 2001, Heraklion, Crete, Greece}, pages 453--461, 2001.
\newblock \href {https://doi.org/10.1145/380752.380839}
  {\path{doi:10.1145/380752.380839}}.

\bibitem[Val82]{ValiantHypercube}
Leslie~G. Valiant.
\newblock A scheme for fast parallel communication.
\newblock {\em SIAM J. Comput.}, 11:350--361, 1982.

\bibitem[VB81]{ValiantBrebner}
L.~G. Valiant and G.~J. Brebner.
\newblock Universal schemes for parallel communication.
\newblock In {\em Proceedings of the Thirteenth Annual ACM Symposium on Theory
  of Computing}, STOC '81, page 263–277, New York, NY, USA, 1981. Association
  for Computing Machinery.
\newblock \href {https://doi.org/10.1145/800076.802479}
  {\path{doi:10.1145/800076.802479}}.

\bibitem[VMS04]{vartiainen2004efficient}
Juha~J Vartiainen, Mikko M{\"o}tt{\"o}nen, and Martti~M Salomaa.
\newblock Efficient decomposition of quantum gates.
\newblock {\em Physical review letters}, 92(17):177902, 2004.

\bibitem[WHY{\etalchar{+}}19]{wu2019optimization}
Bujiao Wu, Xiaoyu He, Shuai Yang, Lifu Shou, Guojing Tian, Jialin Zhang, and
  Xiaoming Sun.
\newblock Optimization of cnot circuits on topological superconducting
  processors.
\newblock {\em arXiv preprint arXiv:1910.14478}, 2019.

\end{thebibliography}
